\documentclass[12pt]{article}

\usepackage{setspace}
\usepackage{mathrsfs}

\usepackage{latexsym}
\usepackage{amsmath}
\usepackage{amsfonts}
\usepackage{amssymb}
\usepackage{amsthm}
\usepackage{epsfig}
\usepackage{subfigure}
\usepackage{arydshln}
\usepackage{algorithm}
\usepackage{algorithmicx}
\usepackage{algpseudocode}
\usepackage{graphicx}
\usepackage{comment}
\usepackage[usenames, dvipsnames]{color}

\usepackage{natbib}
\usepackage{times}
\usepackage[usenames]{color}
\bibpunct{(}{)}{;}{a}{,}{,}

\newcommand{\bX}{\mathbf{X}}

\newcommand{\by}{\mathbf{y}}

\newcommand{\bD}{\boldsymbol{D}}

\newcommand{\ba}{ \color{black} }

\newtheorem{definition}{Definition}
\newtheorem{theorem}{Theorem}

\title{\bf Differentially Private Significance Tests for Regression Coefficients}
\author{Andr\'es F. Barrientos, Jerome P. Reiter, Ashwin Machanavajjhala, Yan Chen \footnote{Andr\'es F. Barrientos is Postdoctoral Associate, Department of Statistical Science, Duke University, Durham, NC 27708 (email: afb26@stat.duke.edu); Jerome Reiter is Professor, Department of Statistical Science, Duke University, Durham, NC 27708 (jerry@stat.duke.edu); Ashwin Machanavajjhala is Assistant Professor, Department of Computer Science, Duke University, Durham, NC 27708 (ashwin@cs.duke.edu); and, Yan Chen is Graduate Student, Department of Computer Science, Duke University, Durham, NC 27708 (yanchen@cs.duke.edu).}}

\begin{document}

\date{}
\maketitle

\begin{abstract}
Many data producers seek to provide users access to confidential data without unduly compromising data subjects' privacy and confidentiality.  One general strategy is to require users to do analyses 
without seeing the confidential data; for example, analysts only get access to synthetic data or query systems that provide disclosure-protected outputs of statistical models.  With synthetic data or redacted outputs, the analyst never really knows how much to trust the resulting findings. In particular, if the user did the same analysis on the confidential data, would regression coefficients of interest be statistically significant or not?  We present  algorithms for assessing this question that satisfy differential privacy.  We describe conditions under which the algorithms should give accurate answers about statistical significance.  We illustrate the properties of the proposed methods using artificial and genuine data.

\end{abstract}

\section{Introduction}\label{Intro}

In many settings, data producers such as national statistical agencies, survey organizations, health systems, and private sector companies---henceforth all called agencies---seek to provide researchers and the broader public access to data on individual records. However, these agencies are ethically and often legally obligated to protect the confidentiality of data subjects' identities and sensitive attributes.  Research has shown that stripping obvious identifiers, like names and addresses, may not suffice to protect confidentiality  \citep[e.g.,][]{sweeney;1997,sweeney;2013, narayanan;shmatikov;2008,parry;chase;2011}.  Ill-intentioned users---henceforth called intruders---may be able to learn sensitive information by linking released data files to records in external databases by matching on fields common to both datasets, thereby breaking the protection from de-identification.

In recognition of this threat, agencies have developed and deployed techniques that allow users to do analyses without seeing the actual data.  One approach is to use remote access query systems \citep{statsciserverkarr} in which the user submits a query  to a server that holds the data for output from some statistical model.  The server runs the query and reports back the analysis results to the user, e.g., estimated regression coefficients and their standard errors, without ever allowing the user to see the individual-level data.  To further reduce disclosure risks, the outputs usually are coarsened or perturbed \citep{okeefe1}; for example, the server can add noise to outputs that satisfies the risk criterion differential privacy \citep{dwork;mcsherry;nissim;smith;2006,wasserman;zhou;2010}. Query system approaches are used by many government agencies, such as the Census Bureau and the Australian Bureau of Statistics,  and are being implemented for general social science data access, for example in DataVerse \citep{king;2007,crosas;2011} and in the Private data Sharing Interface \citep{gaboardi;honake;king;nissim;ullman;vadham;2016}.  A second approach is to release fully synthetic data \citep{rubin;93,raghunathan;reiter;rubin;2003,reiter;2005b, drechslerbook}.   Here, the agency generates new values for every confidential datum by sampling from a predictive distribution estimated with the confidential data.  Since all values are simulated, it is nonsensical for intruders to match released cases to external records.  Synthetic data have been used in several public use data products, including the Longitudinal Business Database  \citep{kinney;reiter;reznek;miranda;jarmin;abowd;2011},  the Survey of Income and Program Participation  \citep{abowd;stinson;benedetto;2006}, and the OnTheMap application \citep{machanavajjhala;kifer;abowd;gehrke;vilhuber;2008}.

While query systems and synthetic data are appealing options for data release, they have a significant drawback: it is difficult for analysts to know how much they should trust the results of their analyses.  For example, in a query system, the user might ask for outputs from a regression model that, in actuality, fits poorly on the confidential data. This lack of fit cannot be easily detected from the coefficients and standard errors alone, whether they are perturbed or not.  Additionally, the steps taken to perturb the outputs could infuse substantial error into the reported coefficients.  Similar dilemmas arise for synthetic data.  By default, the synthetic data reflect only those distributional features and relationships encoded in the synthesis models  \citep{reiter;2005b}. The synthesis models may fail to describe the data in ways that lead the analyst to findings that are not supported by the confidential data.  Further, even when the synthesis models adequately describe the distributions in the confidential data, the process of generating synthetic data tends to increase standard errors, which could obscure important relationships.

The literature on privacy-preserving data analysis has begun to address aspects of this problem.
\cite{reiter;2003}  suggests that linear regression output from query systems be accompanied by synthetic plots of residuals versus predicted values.  Related residual diagnostics for logistic regressions are proposed in \citet{reiter;kohnen;2005} and \citet{o2009regression}. \cite{chen;machanavajjahala;reiter;barrientos;2016}  present an algorithm for releasing residual plots (and also an algorithm for ROC curves for logistic regression) that satisfies differential privacy. Their algorithm takes as input the privately-estimated coefficients, which could come from noisy outputs or a synthetic data analysis.  While useful diagnostic tools, these plots do not provide analysts with means to compare inferences obtained via the privacy-preserving mechanism to those that would be obtained from the confidential data.  It may be, for example, that a particular regression coefficient of substantive interest has a large p-value in the noisy output or synthetic data, even though it has a small one in the confidential data, or vice versa.  

In this article, we present algorithms for comparing the sign and significance level of privately-computed regression coefficients, i.e., those computed via output perturbation or via synthetic data, with those computed from the confidential data.    We envision the outputs of these algorithms being delivered to users via a verification server \citep{reiter;oganian;karr;2009}.  This is a query system that allows users to ask for measures indicating how similar privately-computed results are to those based on the confidential data without allowing users to see the confidential data.  The algorithms satisfy differential privacy, which has important benefits in this interactive context.  As shown in \cite{reiter;oganian;karr;2009} and \cite{mcclure;reiter;2012}, when verification servers provide exact (unperturbed) answers to queries about similarity of results, intruders can query the server repeatedly to gather information that, in combination, provides unacceptably tight ranges for individual confidential values. Differential privacy provides provable bounds for the amount of information leaked by the server over repeated queries, regardless of their nature. 

The basic idea of the algorithms is built on the subsample and aggregate mechanism of \cite{nissim;raskhodnikova;smith;2007}.  We randomly partition the confidential data into $M$ disjoint subsets. In each subset, we estimate the regression using only the data in that subset, from which we compute the univariate $t$-statistic for the regression coefficient(s) of interest to the user.  We  truncate each $t$-statistic at some user-defined threshold $a$; this facilitates  differentially private algorithm design, as we discuss later. We add noise to the average of the truncated $t$-statistics, sampled from a Laplace distribution with variance tuned to satisfy differential privacy.  We refer the resulting noisy statistic to an appropriate reference distribution under the null hypothesis that the coefficient equals zero, resulting in calibrated p-values.  The p-value can be used directly as evidence of the significance of the coefficient, or it can be compared with the corresponding, privately-computed p-value for purposes of verification.  The sign of the noisy  $t$-statistic also provides a differentially private estimate of the sign of the coefficient.

We are not aware of algorithms for differentially private significance tests for linear regression coefficients, although the literature on differential privacy includes significance tests for models appropriate for other settings.  Several authors have developed differentially private significance tests for categorical data and contingency tables  \citep[e.g.,][]{vu;slavkovic;2009,gaboardi;lim;rogers;vadhan;2016,wang;jaewoo;kifer}.  These tests cannot be sensibly used for linear regression, where the outcome variable, as well as potentially some of the explanatory variables, are assumed to be continuous rather than categorical. 
\cite{solea;2014} and \cite{dorazio;honaker;king;2015} propose differentially private significance tests for the mean and the difference of means of Gaussian random variables, respectively. 
These authors assume that bounds for the means or the data values are known, whereas we work in multivariate regression settings where bounds on the variables need not be known.
\cite{campbell;2018;DP_ANOVA} propose a differentially private algorithm for analysis of variance, which is a special case of linear regression with only categorical explanatory variables.  They restrict results to outcomes that lie on the unit interval, do not consider continuous predictors, and report only the result of the omnibus significance test that all coefficients simultaneously equal zero. \cite{karwa;vadhan;2017,karwa;vadhan;2018} propose differentially private algorithms to obtain confidence intervals for single means of Gaussian random variables, which can be inverted to significance tests.  Their approach relies on spending some privacy budget to bound the range and standard deviation of the data values for the single variable.  Their approach does not apply in a straightforward manner to regression modeling with multiple explanatory variables, as the expression for the variance of any estimated regression coefficient has a numerator and denominator that are non-linear functions of all the variables used in the regression model.

% without assuming that the domain of the samples is bounded. \citeauthor{karwa;vadhan;2018} confidence intervals represent an inferential tool to perform hypothesis testing for single means.
% authors define a differential private estimator for the mean (or the difference) and use the corresponding sampling distribution to perform hypothesis testing. 
%Although these works are important initial steps towards developing differential private hypothesis tests, much more work remains to be done in this area, particularly in the regression context. 
%We want to emphasize that the tests proposed in these works cannot be adapted---at least in a straightforward manner---to perform hypothesis testing of regression coefficients. 

Multiple authors have developed differentially private algorithms for estimating pieces of the outputs needed for significance testing of coefficients in linear regressions or other predictive models    
%Some of them include  log-linear models 
\citep[e.g.,][]{sarlos;2006, chaudhuri2009privacy,dwork;lei;2009, chaudhuri2011differentially,kifer;smith;thakurta;2012,zhang;zhang;xiao;yang;winslett;2012,bassily;smith;thakurta;2014, karwa2015private, wu2015revisiting, honkela;das;dikmen;kaski;2016, jing;charest;slavkovic;smith;fienberg;2018}.
%, logistic regression \citep{chaudhuri2009privacy}, ``{\it not strictly differential private}'' linear mixed models \citep{abowd2013differential}, linear models \citep{sheffet;2015,dwork;lei;2009}, among others. 
%Particularly for linear models, algorithms are designed based on different techniques, including output perturbation \citep{chaudhuri2011differentially,wu2015revisiting}, objective    perturbation \citep{chaudhuri2011differentially,kifer;smith;thakurta;2012}, functional mechanism \citep{zhang;zhang;xiao;yang;winslett;2012}, projections \citep{sarlos;2006,sheffet;2015,honkela;das;dikmen;kaski;2016}, gradient descent algorithms \citep{bassily;smith;thakurta;2014}, and trimmed summaries \citep{dwork;lei;2009}. 
None of these works includes procedures to compute standard errors, making it impossible to conduct significance tests.
%\cite{karwa2015private} uses a Bayesian approach for estimating differentially private log-linear models.
%and \cite{abowd2013differential} work under a Bayesian approach that allows them to assess uncertainty of their estimates. However, \citeauthor{karwa2015private}'s (\citeyear{karwa2015private}) proposal only 
%differentially log-linear models 
%and \cite{abowd2013differential} provide non-strictly differential private algorithms.} 
% of the regression coefficients. Hence, joint comparisons---for the sign and significance---between any of these procedures and our proposal are not possible. 
\citet{sheffet;2015} presents an algorithm for estimating regressions that does provide standard errors and hence significance tests; however, the algorithm 
%This $(\epsilon,\delta)$-differential private method is based on the Johnson-Lindenstrauss projection of $A=[X;Y]$, where $X$ denotes the design matrix and $Y$ is a vector containing the observed responses. The Sheffet's method randomly returns one of two possible outputs. The probability of returning either of those outputs depends on the smallest singular value of $A$, say $\lambda(A)$. In particular, if $\lambda(A)$ is large, the probability of returning private $t$-values obtained from a noisy projection of $A$ increases. On the other hand, smaller values of $\lambda(A)$ lead to an increased probability of returning private $t$-values from a rigde-type regression with penalty term $w$, where $w$ is a function of $(\epsilon,\delta,r,B)$, $r$ defines the dimension of the space where $A$ is projected, and $B$ is a bound on the $l_2$-norm of any row in $A$. Hence, under the Sheffet's method, a linear combination of the columns of $X$ that predict $Y$ with high precision will lead to a high probability of obtaining an output from a rigde-type regression. This is a clear limitation for those users interested only in one of these possible outputs.  
%We can think of the Sheffet's algorithm as a direct competitor of our proposal. However, there are different reasons why comparisons between the Sheffet's algorithm and ours are not straightforward to make. First, 
%the Sheffet's algorithm 
sometimes returns output associated with a variant of ridge regression rather than strictly linear regression. The algorithm also requires 
%Second, comparisons can only be made in scenarios where 
all data values to be bounded, which we do not require in our algorithms. Finally, the algorithm appears to require $(\epsilon, \delta > 0)$-differential privacy to give useful outputs, whereas our algorithm allows for $\epsilon$-differential privacy,
%The algorithm in \citet{sheffet;2015} also requires selection of various input parameters, making it complicated to use as a comparison method. %Third, these algorithms use different input parameters. Therefor it is not clear which values should be used to fix the input parameters and ensure a fair comparison. For these reasons, we do not provide comparisons between the Sheffet's algorithm and our proposal in this paper.

%In another related work, \cite{jing;charest;slavkovic;smith;fienberg;2018} present differentially private algorithms for model selection in multivariate linear regression. Their proposal does not rely on significance tests, but on penalized likelihoods. To control the sensitivity of the considered penalized likelihoods, their algorithms require not only to assume that data entries are bounded, but also that the parameter space is constrained. Specifically, the algorithms imposes an $l_1$-norm constraint on the estimated regression coefficient.

The remainder of the article is organized as follows. In Section \ref{Background}, we review differential privacy and some of the techniques used to design algorithms that satisfy it. In Section \ref{Method_dev}, we present the algorithm for the differentially private $t$-statistic, including its reference distribution. In Section \ref{sec:theory}, we discuss some theoretical aspects of the differentially private $t$-statistic, focusing on approximation properties and type II error rates. In Section \ref{Illustrations}, we present results of simulation studies that illustrate the performance of the differentially private $t$-statistic in finite samples.  In Section \ref{ChoosingParams}, we present an approach for choosing the number of partitions and the threshold level.  These drive the accuracy and usefulness of the inferences.  In Section \ref{conclusion}, we conclude with suggestions for implementation of these techniques, as well as discuss future research topics around verification of privately-computed regression quantities.

\section{Review of Differential Privacy}\label{Background}

Before reviewing differential privacy, we motivate why one should not simply release verification measures without redaction. 
%Suppose that an agency has released synthetic data, from which a user estimates the intercept and slope in a regression of an outcome $y$ on a single predictor $x$. 
Suppose that an intruder asks a verification server to return the value of the $t$-statistic for the slope in a regression of an outcome $y$ on a single predictor $x$, and that the verification server provides this value from the regression estimated with the confidential data.
Consider a worst case scenario: the intruder knows the values of $(x_i, y_i)$ for all but one record in the confidential data.  If the intruder submits a regression involving all records and then requests the $t$-statistic, the user can try various combinations of $(x, y)$ for that unknown record until finding the set of values that yield the reported $t$-statistic.  More generally, similar attacks work for an intruder who knows the values of $(x_i, y_i)$ for any $r$ records in the confidential data when the intruder can request output from a regression estimated with those $r$ records plus one additional record.

% An ill-intentioned user could submit two verification requests, one with a particular individual and one without that individual, e.g., by subsetting purposefully on $x$.  
%Differences in the two $t$-statistics reveal information about the excluded individual.  For a simple example, when the $t$-statistic changes from a large negative value to a large positive value, the user might infer that the excluded individual's $y_i$  is a large value (when $x_i$ is also large), thereby potentially leaking sensitive information about this individual. As another example, suppose a worst case scenario: the user knows the values of $(x_i, y_i)$ for $r$ records.  If the user can submit a regression involving data comprising those $r$ records plus one additional record, and then request the $t$-statistic for the slope of that regression, the user can try various combinations of $(x_i, y_i)$ for that unknown record until finding the set of values that yield the reported $t$-statistic.  

% the user can solve for the unknown user can use trial and error able to perfectly guess the wage of such individual. Perfect guesses can occur when statistical summaries are not properly released. The $t$-statistic is not the exception. 

As these examples illustrate, it is desirable to redact the verification measures before releasing them, which we do using differential privacy.
%  For this reason, we use differential privacy principles and elements to design robust releasing mechanisms for $t$-statistics that allow controlling for multiple queries and quantifying the leakage of information. 
%Differential privacy 
%\citep{dwork;mcsherry;nissim;smith;2006} 
%has become a gold standard privacy definition due to its provable guarantees and desirable self-composition properties.
%FELIPE:  CHANGE THE NOTATION TO $\mathcal{A}$ rather than $\mathbf{M}$ since we use $M$ later on.
Let $\mathcal{A}$ be an algorithm that takes as input a database $\mathbf{D}$ and outputs some quantity $o$, i.e., $\mathcal{A}(\mathbf{D}) = o$.  In our context, these outputs are used to form verification measures for the $t$-statistic and sign.  Define neighboring databases, $\mathbf{D}$ and $\mathbf{D}'$, as databases that differ in one row and are identical for all other rows.
Specifically, $\mathbf{D}$ and $\mathbf{D}'$ are neighboring databases if there exists only one record $d \in \mathbf{D}$ and one record $d' \in \mathbf{D}'$ such that $d \neq d'$ and $\mathbf{D} - \{d\} = \mathbf{D}' - \{d'\}$. 
%Let $\mathcal{A}(D_1)$ and $\mathcal{A}(D_2)$ be the output of some algorithm $\mathcal{A}$ applied on $\mathbf{D}_1$ and $\mathbf{D}_2$, respectively.  
% $\epsilon$-differential privacy when the distributions of its 
%Differential privacy \cite{C.Dwork TCC06} was proposed over one decade ago and soon becomes the gold standard privacy definition due to its provable guarantees and good self-composition properties. We define the neighboring relation $\mathbf{Q}$ between any two databases $\mathbf{D}$ and $\mathbf{D}'$ as follows: $(\mathbf{D}, \mathbf{D}') \in \mathbf{Q}$ if $\mathbf{D}$ and $\mathbf{D}'$ only differs by one record. That is there only exists one record $x \in D$ and one record $y \in D'$ such that $x \neq y$ and $\mathbf{D} - \{x\} = D' - \{y\}$.
%An algorithm satisfies differential privacy if its 
%outputs are similar for any two neighboring databases.

\begin{definition}[$\epsilon$-differential privacy]
%Let $\mathcal{A}$ be an algorithm that takes as input a database $\mathbf{D}$ and outputs some quantity $q$, i.e., $\mathcal{A}(D) = q$. 
An algorithm $\mathcal{A}$ satisfies $\epsilon$-differential privacy if for any pair of neighboring databases $(\mathbf{D}, \mathbf{D}')$, and any non-negligible measurable set $S \subseteq range(\mathcal{A})$,
%\begin{eqnarray}
the $Pr(\mathcal{A}(\mathbf{D}) \in S) \leq \exp(\epsilon) Pr(\mathcal{A}(\mathbf{D}') \in S).$
%\end{eqnarray}
\end{definition}
Intuitively, $\mathcal{A}$ satisfies $\epsilon$-DP when the distributions of its 
%Differential privacy \cite{C.Dwork TCC06} was proposed over one decade ago and soon becomes the gold standard privacy definition due to its provable guarantees and good self-composition properties. We define the neighboring relation $\mathbf{Q}$ between any two databases $\mathbf{D}$ and $\mathbf{D}'$ as follows: $(\mathbf{D}, \mathbf{D}') \in \mathbf{Q}$ if $\mathbf{D}$ and $\mathbf{D}'$ only differs by one record. That is there only exists one record $x \in D$ and one record $y \in D'$ such that $x \neq y$ and $\mathbf{D} - \{x\} = D' - \{y\}$.
%An algorithm satisfies differential privacy if its 
outputs are similar for any two neighboring databases, where similarity is defined by the factor $\exp(\epsilon)$.   The $\epsilon$, also known as the privacy budget, controls the degree of the privacy offered by $\mathcal{A}$, with lower values implying greater privacy guarantees. $\epsilon$-DP is a strong criterion, since
even an intruder who has access to all of $\mathbf{D}$ except any one row learns little from $\mathcal{A}(\mathbf{D})$  about the values in that unknown row when $\epsilon$ is small.

Differential privacy has three other properties that are appealing for verification measures.  
%The following composition properties hold for any differentially private algorithms, which are usually used to prove or design complex differentially private algori%thms. 
Let $\mathcal{A}_1(\cdot)$ and $\mathcal{A}_2(\cdot)$ be $\epsilon_1$-DP and $\epsilon_2$-DP algorithms.  First, 
for any database $\mathbf{D}$, releasing the outputs of both $\mathcal{A}_1(\mathbf{D})$ and $\mathcal{A}_2(\mathbf{D})$ ensures $(\epsilon_1 + \epsilon_2)$-DP.  Thus, we can quantify and track the total privacy leakage from releasing verification measures.  Second, releasing the outputs of both $\mathcal{A}_1(\mathbf{D}_1)$ and $\mathcal{A}_2(\mathbf{D}_2)$, where $\mathbf{D}_1 \cap \mathbf{D}_2 = \emptyset$, satisfies $\max \{\epsilon_1, \epsilon_2\}$-DP.  Third, for any algorithm $\mathcal{A}_3(\cdot)$, releasing $\mathcal{A}_3(\mathcal{A}_1(\mathbf{D}))$ for any $\mathbf{D}$ still ensures $\epsilon_1$-DP.  Thus, post-processing the output of $\epsilon$-DP algorithms does not incur extra loss of privacy.

%\begin{itemize}
%\item[1.] Sequential~Composition: Releasing the outputs of both $\mathbf{M_1}(D)$ and $\mathbf{M_2}(D)$ on any database $\mathbf{D}$ ensures $\epsilon_1 + \epsilon_2$-differ%ential privacy.
%\item[2.] Parallel~Composition: Releasing the outputs of both $\mathbf{M_1}(D_1)$ and $\mathbf{M_2}(D_2)$, where $\mathbf{D}_1 \cap D_2 = \emptyset$, satisfies $\max \{\epsi%lon_1, \epsilon_2\}$-differential privacy
%\item[3.] Post-processing: For any algorithm $\mathbf{M_3}(\dot)$, releasing $\mathbf{M_3}(\mathbf{M_1}(D))$ on any database $\mathbf{D}$ still ensures $\epsilon_1$-differen%tial privacy, which means any post-processing on the output of differential privacy algorithms will not incur extra loss of privacy.
%\end{itemize}

%FELIPE:  DECIDE WHICH OF THE TWO PARAGRAPHS ABOUT  LAPLACE MECHANISM YOU WANT.

A common method for ensuring $\epsilon$-DP is the Laplace Mechanism \citep{dwork;mcsherry;nissim;smith;2006}.
%which adds noise drawn from the Laplace distribution with a scale proportional to the global 
%sensitivity the output aggregates.
For any function $f : \mathbf{D} \rightarrow \mathbb{R}^d$, let $\Delta(f) = \max_{(\mathbf{D_1},\mathbf{D_2})} ||f(\mathbf{D}_1) - f(\mathbf{D}_2)||_{1}$, where $(\mathbf{D}_1, \mathbf{D}_2)$ are 
neighboring databases.  This quantity, known as the global sensitivity of $f$, is 
 %\begin{definition}[Global Sensitivity]
%The global sensitivity of a function $f : D \rightarrow \mathbf{R}^n$, denoted as $\Delta(f)$, is defined to be 
the maximum $L_1$ distance of the outputs of the function $f$ between any two neighboring databases.
%\begin{eqnarray}
%\Delta(f) = \max_{(\mathbf{D_1},\mathbf{D_2})} ||f(\mathbf{D}_1) - f(\mathbf{D}_2)||_{1}.
%\end{eqnarray}
%\end{definition}
%
%\begin{definition}[Laplace Mechanism]
%For any function $f : D \rightarrow \mathbf{R}^n$, 
The Laplace Mechanism is 
\begin{equation}\label{def_LM}
\mathbf{LM}(\mathbf{D}) = f(\mathbf{D}) + \eta,
\end{equation}
where $\eta$ is a $d \times 1$ vector of 
independent draws from a Laplace distribution with density $p(x \mid \lambda) = (1/(2\lambda)) \exp(-|x| / \lambda)$, where $\lambda = \Delta(f) / \epsilon$.
We use the Laplace Mechanism to design verification measures that satisfy $\epsilon$-differential privacy, which we refer to as 
 $\epsilon$-DP verification measures.

We also use the subsample and aggregate technique \citep{nissim;raskhodnikova;smith;2007}. This technique allows us to reduce the global sensitivity of $f$, thereby reducing the variance in the noise distribution. To implement this technique, we randomly partition the dataset $\bD$ into $M$ disjoint subsets, $\bD_1,\ldots,\bD_M$. We then compute $f(\bD_1),\ldots,f(\bD_M)$ and their average $M^{-1}\sum_{l=1}^M f(\bD_l)$.  The global sensitivity of $M^{-1}\sum_{l=1}^M f(\bD_l)$ is $1/M$ times that of $f(\bD)$, since any single observation appears in at most one of the partitions.  Finally, we use the Laplace mechanism to release a noisy version of $M^{-1}\sum_{l=1}^M f(\bD_l)$.  

\ba

\section{The Differentially Private Test Statistic}\label{Method_dev}

We begin by laying out relevant notation and formally specifying our objectives. Let $\bD$ be a confidential dataset comprising $n$ individuals. For each individual $i=1, \dots, n$, let $y_i\in \mathbb{R}$ be its univariate response variable and $x_i = (1,x_{i,1},\ldots,x_{i,p})^\top \in \mathbb{R}^{p+1}$ be its $(p+1) \times 1$ vector of predictors.  Hence,  $\bD =\{(x_i, y_i)\}_{i=1}^n$.  An analyst seeks to estimate the parameters in the regression, $y_i =  \boldsymbol{\beta}^\top x_i + e_i$, where $\boldsymbol{\beta} = (\beta_0,\ldots,\beta_p)^\top \in \mathbb{R}^{p+1}$ and $e_i$ are i.i.d. random errors with $E(e_i)=0$ and $Var(e_i)=\sigma^2$. In linear regression, we typically assume that $e_i \sim N(0, \sigma^2)$ for all $i$, although our algorithm can be used with other error distributions.  We assume that, if the analyst had direct access to $\bD$, he or she would make inferences about each $\beta_j$ based on the maximum likelihood estimator (MLE), $\hat{\beta_j}$, and its corresponding sampling distribution.  However, the analyst does not get direct access to $\bD$; instead, the analyst can learn only the privately-computed estimates $\tilde{\beta_j}$, which could arise from perturbed versions of $\hat{\beta}_j$ or from synthetic data.  
Our key question is whether or not inferences about $\beta_j$ are similar when using $\hat{\beta}_j$ or $\tilde{\beta_j}$.

To address this question, we develop differentially private significance tests.  Let $T(\bD)$ be the standardized estimator of $\beta_j$ obtained from $\bD$, that is, 
$$
T(\bD) = \hat{\beta_j}/ \sqrt{\hat{\Sigma}_{j,j}},
$$
where $\left(\hat{\Sigma}_{j,j}\right)$ is the $(j,j)$th element of the matrix, $\hat{\Sigma} = \hat{\sigma}^2(\bX_{\bD}^\top\bX_{\bD})^{-1}$. Here,  $\hat{\sigma}^2 = (\by_{\bD}-\hat{\boldsymbol{\beta}}^\top\bX_{\bD})^\top(\by_{\bD}-\hat{\boldsymbol{\beta}}^\top\bX_{\bD})/(n-p-1)$, where $\by_{\bD}=(y_1,\ldots,y_n)^\top$, and $\bX_{\bD}=[x_1^\top,\ldots,x_n^\top]^\top$ is the design matrix associated with the regression when estimated with $\bD$.  Under certain conditions on $\bX_{\bD}$, $\hat{\beta_j}$  is asymptotically normally distributed for a large range of error distributions (\citeauthor{vandervaart;2000}, \citeyear{vandervaart;2000}, page 21;  \citeauthor{bhattacharya;lin;patrangenaru;2016}, \citeyear{bhattacharya;lin;patrangenaru;2016}, section 6.8).
The key condition on the design is that the maximum among the diagonal elements of the matrix $\bX_{\bD}(\bX_{\bD}^\top \bX_{\bD})^{-1}\bX_{\bD}^\top$ goes to zero as $n \rightarrow \infty$.  Hence, for a large enough $n$, the distribution of $T(\bD)$ can be suitably approximated by a standard Gaussian distribution. In the remainder of the article, we refer to $T(\bD)$ simply as the $t$-statistic.

$T(\bD)$ provides all the information needed for inferences about the sign and significance of $\beta_j$.  Hence, we can address our key question and account for privacy by developing algorithms for releasing differentially private versions of $T(\bD)$, along with deriving reference distributions for the private $t$-statistics.  Taken together, these algorithms enable the analyst to assess the significance level for $\beta_j$ directly from the private output.  

Unfortunately, we cannot simply apply the Laplace mechanism in (\ref{def_LM}) to create the differentially private test statistic, as the global sensitivity of $T(\bD)$ is unbounded.  A possible remedy  is to work with some bounded statistic instead of $T(\bD)$.  For statistical significance, two obvious candidates include (i) the p-value associated with $T(\bD)$ and (ii) a truncated version of $T(\bD)$.  We next describe some of the pros and cons of each approach.

Let $p^T$ be the p-value associated with $T(\bD)$ for a two-tailed significance test of the null hypothesis $\beta_j=0$. Any $p^T$ has global sensitivity equal to one. Let $p^{T,\epsilon}$ be the $\epsilon$-differentially private p-value obtained after adding Laplace noise to $p^T$ based on the global sensitivity of one. With high probability, adding this noise to small values of $p^T$ could inflate them so much as to change our opinion of the significance of $\beta_j$. This is less problematic when adding noise to large values of $p^T$. In other words, for a given $\epsilon$, the probability that an analyst reaches the same decisions about statistical significance when using $p^{T,\epsilon}$ or $p^{T}$ is higher when $p^{T}$ falls in an acceptance region for $H_0$ than when $p^{T}$ falls in a rejection region. 

Regarding the second approach, let the truncated $t$-statistic be given by   
\begin{eqnarray}\label{truncatedt}
T^t(\boldsymbol{D}) 
=
\left\{
\begin{array}{ll}
-a & \mbox{if } T(\bD) < -a,\\
T(\bD) &  \mbox{if } -a \leq T(\bD) \leq a,\\
 a &  \mbox{if }T(\bD) > a, \\
\end{array}
\right.
\end{eqnarray}
where $a>0$ is a user-defined parameter. Here,  $a$ has to be large enough to ensure that $T(\boldsymbol{D})$ and $T^t(\boldsymbol{D})$ lead to the same conclusion regarding the null hypothesis with high probability.  Because of the truncation, the global sensitivity of $T^t(\boldsymbol{D})$ equals $2a$.  Let  $T^{t,\epsilon}(\boldsymbol{D})  = T^t(\boldsymbol{D})  + \eta$ be a noisy version of $T^t(\boldsymbol{D})$, where $\eta \sim \mathrm{Lap}(0,2a/\epsilon)$ and $\mathrm{Lap}(l,s)$ denotes the Laplace distribution with location $l$ and scale $s$. 
The problematic situations for $T^{t,\epsilon}(\boldsymbol{D})$ are the reverse of those for $p^T$.
With undesirably high probability, adding noise to values of $T^t(\boldsymbol{D})$ near zero could make an insignificant effect appear significant, whereas the noise is not likely to change our opinion about significance when $T(\boldsymbol{D})$ is large. Put another way, 
for a given $\epsilon$, the probability that an analyst reaches the same decisions about statistical significance when using $T^{t,\epsilon}(\boldsymbol{D})$ or $T^{t}(\boldsymbol{D})$ is higher when $T^{t}(\boldsymbol{D})$ falls in a rejection region for $H_0$ than when $T^{t}(\boldsymbol{D})$ falls in an acceptance region.

%The arguments above suggest that neither approach always outperforms the other.  We opt for the second approach because $T^t(\bD)$ is more analytically tractable than $p^T$.  Using $T^t(\bD)$ also allows the release of a noisy estimate of the sign of $\beta_j$ without additional expenditure of $\epsilon$.
The arguments above suggest that neither approach always outperforms the other.  We opt for the second approach because $T^t(\bD)$ is more analytically tractable than $p^T$. As a result, we find it easier to develop a properly calibrated significance test, and understand its theoretical properties, for $T^t(\bD)$ than for $p^T$.  Using $T^t(\bD)$ also allows the release of a noisy estimate of the sign of $\beta_j$ without additional expenditure of $\epsilon$, which improves the overall utility of the data release without sacrificing privacy.

%\subsection{statistic}\label{methods}

Since the length of the range of $T^t(\bD)$ coincides with its global sensitivity,
% the variance in the Laplace distribution is driven entirely by $\epsilon$. This implies 
we need to use a large $\epsilon$ to ensure that $T^{t,\epsilon}(\boldsymbol{D})$ is practically useful, perhaps larger than what we would like from the perspective of protecting privacy.  Put another way, for a small $\epsilon$ the noise introduced by the Laplace mechanism may be so large compared to $T^t(\bD)$ that the statistic has little ability to detect any deviations from the null hypothesis.  Hence, we need to adapt the truncated $t$-statistic to reduce the global sensitivity.

We consider a way to do so based on the subsample and aggregate method described in Section \ref{Background}. We first randomly partition $\bD$ into $M$ disjoint subsets, $\mathcal{P} = \{\bD_1,\ldots,\bD_M\}$, of equal size (or as close to equal as possible when $n/M$ is not an integer).  In each $\bD_l$, we estimate the regression model of interest using only $\bD_l$. For the regression coefficient of interest, we then compute the set of $M$  $t$-statistics, $\{T(\bD_1), \dots, T(\bD_M)\}$, and truncate each $T(\bD_l)$ at $[-a,a]$, akin to (\ref{truncatedt}). Let $\{T^t(\bD_1), \dots, T^t(\bD_M)\}$ be the set of $M$ truncated $t$-statistics.  We then compute  
\begin{equation*}
\bar{T}^{t}(\mathcal{P}) = \sum_{l=1}^M T^t(\bD_l)/M.
\end{equation*}
The sensitivity of $\bar{T}^{t}(\mathcal{P})$ is $1/M$ times the sensitivity of $T^t(\bD)$, which apparently achieves our goal. However, we do not create the differentially private measure by adding Laplace noise to $\bar{T}^{t}(\mathcal{P})$, for reasons we now describe.  

Because of the random partitioning, it is reasonable to consider each $\bD_l$ as a random sample from a population (with infinite sample size) and, thus, each $T^t(\bD_l)$ as a random draw from its sampling distribution.  We would like the sampling distributions of $T^t(\bD)$ and $\bar{T}^{t}(\mathcal{P})$ to be approximately the 
%same variance under the null and alternative hypotheses.   If the distribution of the adapted $t$-statistic has approximately the 
same, so that the significance test based on $\bar{T}^{t}(\mathcal{P})$ would have approximately the same power function as that based on $T^t(\bD)$.  When this is the case,  analysts should have high probability of reaching similar conclusions  when using the adapted $t$-statistic or $T^t(\bD)$.  However, the variance of $\bar{T}^t(\mathcal{P})$ is roughly $M$ times smaller than the variance of $T^t(\bD)$.  Thus, instead of using $\bar{T}^t(\mathcal{P})$ directly, 
we equate the variances by multiplying $\bar{T}^{t}(\mathcal{P})$  by $\sqrt{M}$; that is, we use $\bar{T}^{t,R}(\mathcal{P})= \sqrt{M}\bar{T}^t(\mathcal{P})$.  This changes the global sensitivity, as it increases from $2a/M$ to $2a/\sqrt{M}$.  Hence, the differentially private version of the $t$-statistic is $\bar{T}^{t,\epsilon}(\mathcal{P})  = \bar{T}^{t,R}(\mathcal{P})  + \eta$, where $\eta \sim \mathrm{Lap}(0,2a/\sqrt{M}\epsilon)$. 

Figure \ref{Figure_algorithms} displays the steps that agencies can use to release $\bar{T}^{t, \epsilon}(\mathcal{P})$ (Algorithm 1). The figure also describes simple Monte Carlo algorithms that can be used to approximate the sampling distribution of  $\bar{T}^{t, \epsilon}(\mathcal{P})$ (Algorithm 2).  This reference distribution can be used to obtain approximate p-values corresponding to the test statistic. The Monte Carlo simulations are needed to properly account for all sources of randomness, including the noise from the Laplace Mechanism. Taken together, Algorithm 1 and 2 provide a means to perform differentially private significance tests. Inferences about the sign of $\beta_j$ can be obtained from the sign of $\bar{T}^{t, \epsilon}(\mathcal{P})$. 

\begin{figure}[t]
\centering
\includegraphics[scale=0.65]{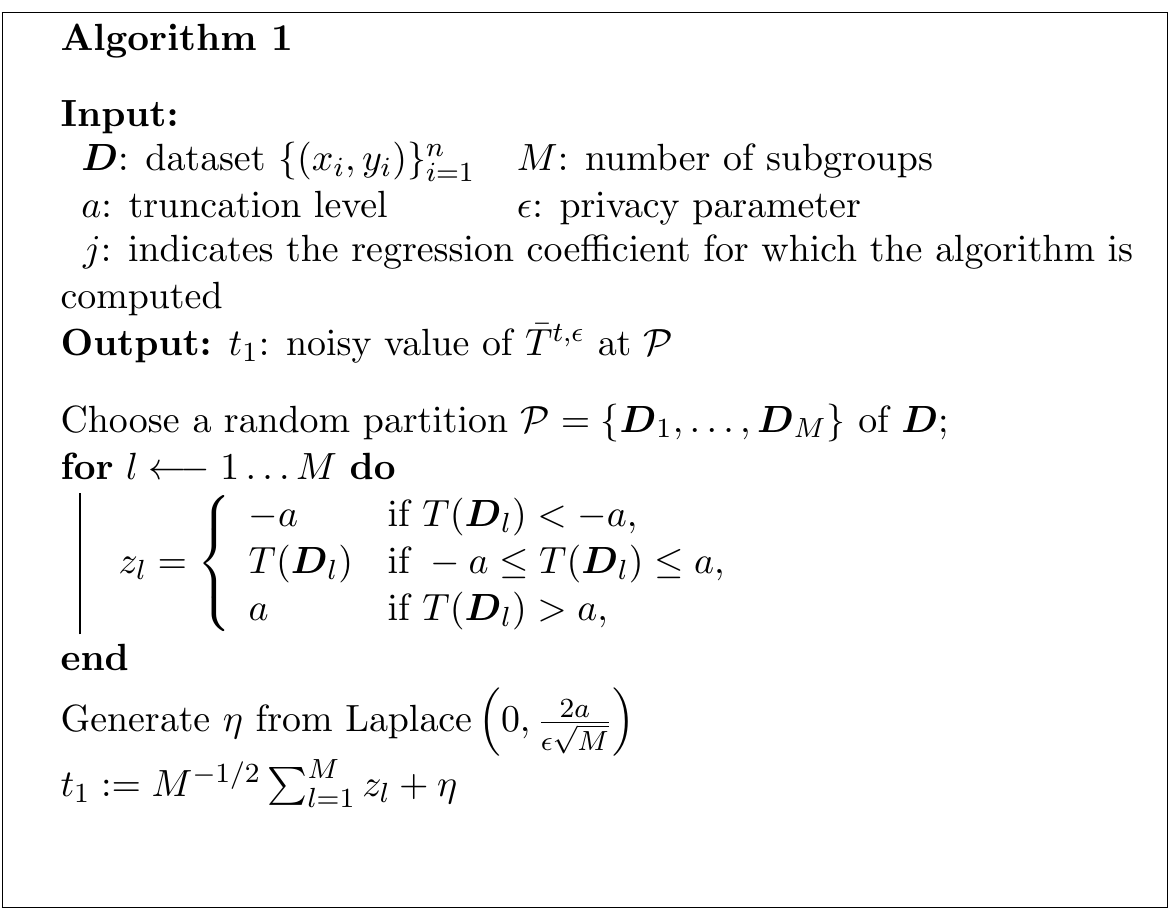}
\includegraphics[scale=0.65]{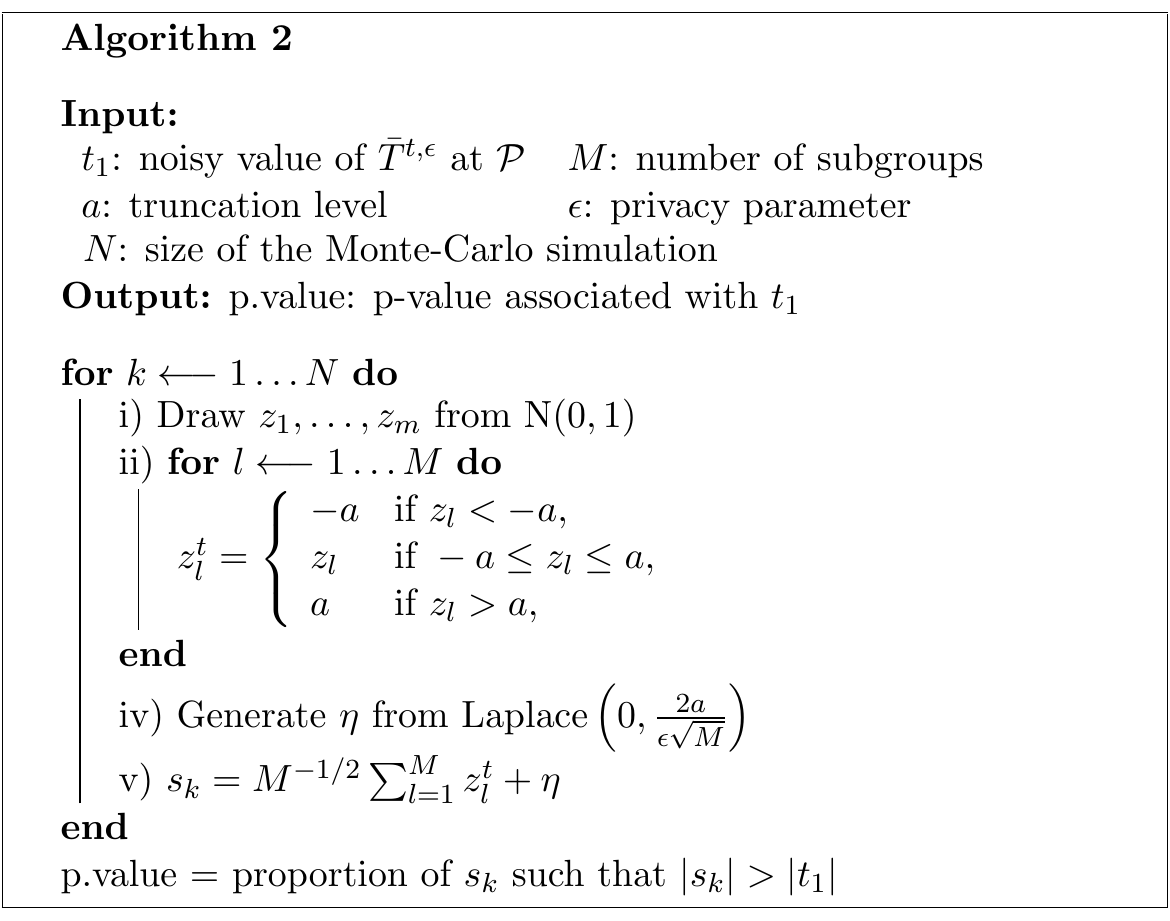}
\caption{Differentially private algorithms for $\bar{T}^{t,\epsilon}$. Algorithm 1 shows the steps to release values of $\bar{T}^{t,\epsilon}$, for a given $(M, a, \epsilon)$. Algorithm 2 shows how to compute the p-value associated with the released $t$-statistics of Algorithm 1.}
\label{Figure_algorithms}
\end{figure}

%\subsection{Proof that $\bar{T}^{t,\epsilon}(\mathcal{P})$ is differentially private}\label{ss:proof}

Finally, we conclude this section with a formal theorem and proof that Algorithm 1 is differentially private. 
\begin{theorem}\label{Algorithms_proof}
Algorithm 1 satisfies $\epsilon$-differential privacy.
\end{theorem}
\begin{proof}
$\bar{T}^{t,R}$  has global sensitivity equal to $2a/\sqrt{M}$. Hence, defining 
$\bar{T}^{t,\epsilon}(\mathcal{P})  = \bar{T}^{t,R}(\mathcal{P}) + \mathrm{Lap}(2a/\sqrt{M}\epsilon)$ implies that, by Definition \ref{def_LM}, Algorithm 1 satisfies $\epsilon$-differential privacy.
\end{proof}

\section{Theoretical Properties}\label{sec:theory}

In this section, we discuss some of the theoretical properties of $\bar{T}^{t,\epsilon}(\mathcal{P})$. First, we derive conditions that characterize the distance between $\bar{T}^{t,\epsilon}(\mathcal{P})$ and $T(\bD)$. We then study the asymptotic probability of type II errors for the statistical test defined from $\bar{T}^{t,\epsilon}(\mathcal{P})$. 
%Third, we provide a formal theorem and proof that $\bar{T}^{t,\epsilon}(\mathcal{P})$ is a differentially private statistic.

\subsection{Distance between $\bar{T}^{t,\epsilon}(\mathcal{P})$ and $T(\bD)$}\label{ss:distance}

To characterize the distance between $\bar{T}^{t,\epsilon}(\mathcal{P})$ and $T(\bD)$, we focus on the probability \linebreak $
P\left\{\left|\bar{T}^{t,\epsilon}(\mathcal{P})- T(\bD)\right|>c |\bX_{\bD},\mathcal{R}\right\}
$ for all $c>0$, where $\mathcal{R}$ denotes the random mechanism used to partition $\bD$ and to create $\mathcal{P}$. We note that $\bD$ and $\mathcal{P}$ are functions of $(\by_{\bD},\bX_{\bD})$ and $(\by_{\bD},\bX_{\bD},\mathcal{R})$, respectively.  Since we are conditioning on $\bX_{\bD}$ and $\mathcal{R}$, the randomness in $T(\bD)$ and $\bar{T}^{t,\epsilon}(\mathcal{P})$ comes from treating $\by_{\bD}$ as a random variable. We use the triangle inequality to bound this probability and, in this way, focus on characterizing each of the terms defining the right hand-side of 
\begin{eqnarray} \nonumber
P\left\{\left|\bar{T}^{t,\epsilon}(\mathcal{P})- T(\bD)\right|>c |\bX_{\bD},\mathcal{R}\right\} & \leq & 
P\left\{\left|\bar{T}^{t,\epsilon}(\mathcal{P})-\bar{T}^{t,R}(\mathcal{P})\right|>c |\bX_{\bD},\mathcal{R}\right\} \\ \nonumber && + \ 
P\left\{\left|\sqrt{M} \bar{T}(\mathcal{P})-T(\bD)\right|>c |\bX_{\bD}, \mathcal{R}\right\} \\ \label{Distance.RHS} && + \
P\left\{\left|\bar{T}^{t,R}(\mathcal{P})-\sqrt{M} \bar{T}(\mathcal{P})\right|>c |\bX_{\bD}, \mathcal{R}\right\},
\end{eqnarray}
where $\bar{T}(\mathcal{P}) = M^{-1} \sum_{l=1}^M T(\bD_l)$. The first term is relatively straightforward to understand theoretically and corresponds to computing a probability under the Laplace distribution. Thus, 
\begin{eqnarray}\label{Prob.RHS1}
P\left\{\left|\bar{T}^{t,\epsilon}(\mathcal{P})-\bar{T}^{t,R}(\mathcal{P})\right|>c |\bX_{\bD},\mathcal{R}\right\} &=& 
\exp\left(-\frac{c\epsilon\sqrt{M}}{2a}\right). 
\end{eqnarray}
As expected, the bound in \eqref{Prob.RHS1} indicates that $\bar{T}^{t,\epsilon}(\mathcal{P})$ gets closer to $\bar{T}^{t,R}(\mathcal{P})$ as the scale of the underlying Laplace distribution, $2a/\epsilon\sqrt{M}$, goes to zero, that is, as $M$ and $\epsilon$ increase and $a$ decreases.  

Turning to the second term in \eqref{Distance.RHS}, we now provide conditions that ensure $\sqrt{M} \bar{T}(\mathcal{P})$ is a reasonable approximation of $T(\bD)$. To begin, we treat each $\bD_l$ as an independent sample from an infinite population.  We make inferences conditional on each $\bX_{\bD_l}$. Throughout, we rely on the following assumption.\\

\noindent {\bf Assumption A1.} {\it  The distribution of $\hat{\beta}_{jl}$, i.e., the MLE of $\beta_j$ estimated with $\bD_l$ treating $\bX_{\bD_l}$ as fixed, can be suitably approximated by a Gaussian distribution with mean $\beta_j$ and variance $\Sigma_{j,j}(\bD_l)$, where $\Sigma_{j,j}(\bD_l)$ denotes the $j$th diagonal element of $\Sigma(\bD_l) = {\sigma}^2(\bX_{\bD_l}^T\bX_{\bD_l})^{-1}$ and $\bX_{\bD_l}$ is the design matrix associated with  $\bD_l$.}\\

{\bf A1} is widely assumed in practice in regression modeling (see \citeauthor{vandervaart;2000}, \citeyear{vandervaart;2000}, page 21;  \citeauthor{bhattacharya;lin;patrangenaru;2016}, \citeyear{bhattacharya;lin;patrangenaru;2016}, section 6.8). 
Under {\bf A1}, treating each $\bD_l$ as independent samples and conditioning on $\bX_{\bD_l}$ implies that $\sqrt{M} \bar{T}(\mathcal{P})$ and $T(\bD)$ are both Gaussian-distributed with variance equal to one.  However, it is not necessarily the case that the means of each $T(\bD_l)$ are equal just because the conditional means of each $\hat{\beta}_{jl}$ are equal, nor that these means equal the mean of $T(\bD)$.  In particular, when $\beta_j \neq 0$,  $T(\bD_l)$ is based on a smaller sample size and a different design matrix than $T(\bD)$, which results in different (typically smaller)  expected values.  Thus, we need to derive conditions on the means of each $\sqrt{M}T(\bD_l)$, where $l=1, \ldots, M$, that guarantee the mean of $\sqrt{M} \bar{T}(\mathcal{P})$ is close to the mean of $T(\bD)$. 

Since each $\bD_l$ in actuality is a random sample from $\bD$, as long as the sample size in each partition is large it is reasonable to assume that $(\bX_{\bD_l}^\top\bX_{\bD_l}) \approx (\bX_{\bD_k}^\top\bX_{\bD_k})$ for all pairs of datasets $(l,k)$.  Moreover, it also is reasonable to make the following assumption.\\

\noindent {\bf Assumption A2.} {\it $M(\bX_{\bD_l}^\top\bX_{\bD_l}) \approx (\bX_{\bD}^\top\bX_{\bD})$ for all $l = 1, \ldots, M$.}\\

With {\bf A2}, we have  $M^{-1}\Sigma(\bD_l) \approx \Sigma(\bD)$.  Using this approximation, we take expectations of $\sqrt{M} \bar{T}(\mathcal{P})$ conditional on the realized $(\bX_{\bD}, \mathcal{R})$. We have 
\begin{align}\nonumber
E\left\{\sqrt{M} \bar{T}(\mathcal{P})|\bX_{\bD}, \mathcal{R}\right\} & = M^{-1} \sum_{l=1}^M E\left\{\sqrt{M} T(\bD_l)|\bX_{\bD}, \mathcal{R}\right\} \\\nonumber
& = \   M^{-1}  \sum_{l=1}^M \sqrt{M} \beta_j/\sqrt{\Sigma_{j,j}(\bD_l)}  \\\label{eq:Expectation}
& \approx \   M^{-1}  \sum_{l=1}^M \beta_j/\sqrt{\Sigma_{j,j}(\bD)}  = \   E\left\{T(\bD)|\bX_{\bD}\right\}.
\end{align}
Additionally, using the approximation in {\bf A2}, for any $\bD_l$ the
$$
Cov\left\{(\bX_{\bD}^\top \bX_{\bD})^{-1} \bX_{\bD}^\top \by_{\bD},(\bX_{\bD_l}^\top \bX_{\bD_l})^{-1} \bX_{\bD_l}^\top \by_{\bD_l} \mid \bX_{\bD}, \mathcal{R}\right\} =
\sigma^2(\bX_{\bD}^\top \bX_{\bD})^{-1} = \Sigma(\bD).
$$
Therefore, we have 
\begin{align}\nonumber
Cov\left\{\left.T(\bD),\sqrt{M}\bar{T}(\mathcal{P})\right|\bX_{\bD}, \mathcal{R}\right\} & = 
 \frac{\sqrt{M}}{M}\sum^M_{l=1}Cov\left\{\left.T(\bD),T(\bD_l)\right|\bX_{\bD}, \mathcal{R}\right\} \\\label{eq:Covariance}
& = 
\frac{1}{M}\sum^M_{l=1}\frac{\sqrt{\Sigma_{j,j}(\bD)}}{\sqrt{M^{-1}\Sigma_{j,j}(\bD_l)}}>0.
\end{align}
Thus, the smaller the distance is between $\bX_{\bD}^\top\bX_{\bD}$ and $M(\bX_{\bD_l}^\top\bX_{\bD_l})$, the higher is the correlation between $\sqrt{M} \bar{T}(\mathcal{P})$ and $T(\bD)$. 

A direct application of the Markov inequality implies that
\begin{eqnarray}\nonumber
c^2P\left\{\left|\sqrt{M} \bar{T}(\mathcal{P}) - T(\bD)\right|>c |\bX_{\bD},\mathcal{R}\right\}  &\leq&
\\\nonumber
&& \hspace*{-75mm} Var\left\{\sqrt{M} \bar{T}(\mathcal{P})|\bX_{\bD},\mathcal{R}\right\} + Var\left\{T(\bD)|\bX_{\bD},\mathcal{R}\right\}   - 2Cov\left\{\left.T(\bD),\sqrt{M}\bar{T}(\mathcal{P})\right|\bX_{\bD},\mathcal{R}\right\} \\\label{Prob.RHS2}
&& \hspace*{-20mm} + \left[E\left\{\left.T(\bD)\right|\bX_{\bD},\mathcal{R}\right\} - E\left\{\left.\sqrt{M}\bar{T}(\mathcal{P})\right|\bX_{\bD},\mathcal{R}\right\}\right]^2. %\label{markov}
\end{eqnarray}
Thus, under {\bf A2},
%the assumption that $\bX_{\bD}^\top\bX_{\bD} \approx M(\bX_{\bD_l}^\top\bX_{\bD_l})$,
by {\bf A1} and \eqref{eq:Expectation} we have 
$$
Var\left\{\sqrt{M} \bar{T}(\mathcal{P})|\bX_{\bD},\mathcal{R}\right\} \approx Var\left\{T(\bD)|\bX_{\bD},\mathcal{R}\right\} \approx 1
$$
and 
$$
E\left\{\left.T(\bD)\right|\bX_{\bD},\mathcal{R}\right\} \approx E\left\{\left.\sqrt{M}\bar{T}(\mathcal{P})\right|\bX_{\bD},\mathcal{R}\right\}.
$$
By \eqref{eq:Covariance}, we have 
$$
Cov\left\{\left.T(\bD),\sqrt{M}\bar{T}(\mathcal{P})\right|\bX_{\bD},\mathcal{R}\right\} \approx 1.
$$
Therefore, the probability in \eqref{Prob.RHS2} is near zero, implying that the distance between $\sqrt{M} \bar{T}(\mathcal{P})$ and $T(\bD)$ has high probability of being small. 

Finally, we provide an upper bound for the last term in \eqref{Distance.RHS}. This bound allows us to understand how choices of $M$ and $a$ affect the distance between $\sqrt{M} \bar{T}(\mathcal{P})$ and $\bar{T}^{t,R}(\mathcal{P})$. Specifically, it follows that
\begin{eqnarray}\nonumber
P\left\{\left|\sqrt{M} \bar{T}(\mathcal{P}) - \bar{T}^{t,R}(\mathcal{P})\right|>c |\bX_{\bD},\mathcal{R}\right\}  &\leq&
P\left\{\left|\sqrt{M} \bar{T}(\mathcal{P}) - \bar{T}^{t,R}(\mathcal{P})\right|>0 |\bX_{\bD},\mathcal{R}\right\} \\\nonumber
& = & 1-P\left\{\sqrt{M} \bar{T}(\mathcal{P}) = \bar{T}^{t,R}(\mathcal{P}) |\bX_{\bD},\mathcal{R}\right\} \\\nonumber
& \leq & 1-\prod_{l=1}^M P\left\{T(\bD_l) = T^t(\bD_l)|\bX_{\bD},\mathcal{R}\right\} \\ \label{Prob.RHS3}
& = & 1-\left(\Phi\left(a-\mu\right) - \Phi\left(-a-\mu\right)\right)^M,
\end{eqnarray}
where $\mu = E\{T(\bD_l)|\bX_{\bD_l}, \mathcal{R}\}$ and $\Phi$ denotes the cumulative distribution function of the standard Gaussian distribution. The probability in \eqref{Prob.RHS3} reveals that $\bar{T}^{t,R}(\mathcal{P})$ gets closer to $\sqrt{M} \bar{T}(\mathcal{P})$ as $a$ increases and $M$ decreases. Together, \eqref{Prob.RHS1}--\eqref{Prob.RHS3} provide a full characterization of \eqref{Distance.RHS}. 
%This completes our first result.

\subsection{Asymptotic power properties of  $\bar{T}^{t,\epsilon}(\mathcal{P})$}\label{ss:power}

We next study the type II error rates for the significance test defined from $\bar{T}^{t,\epsilon}(\mathcal{P})$. 
For given values of $(M, a, \epsilon)$, and under $H_0: \beta_j=0$ (i.e., $\mu=0$), let $r$ be a positive constant such that 
$$
P_{H_0}\left\{|\bar{T}^{t, \epsilon}(\mathcal{P})| < r \left| \bX_{\bD},\mathcal{R} \right. \right\}= 1-\alpha,
$$
where $P_{H_0}$ denotes the probability computed under $H_0$. Here,  $r$ corresponds to the critical value that ensures a significance level of $\alpha$ for the test $\mathbb{T}_{t,\epsilon}(\mathcal{P}) = \mathbb{I}_{(-r,r)}(\bar{T}^{t, \epsilon}(\mathcal{P}))$, where $\mathbb{I}_B(b)=1$ if $b\in B$ and $\mathbb{I}_B(b)=0$ otherwise. Thus, the type II error probability associated with this test is given by 
$E_{H_1}\left\{\mathbb{T}_{t,\epsilon}(\mathcal{P}) \left| \bX_{\bD},\mathcal{R} \right.\right\}$,
where $E_{H_1}$ denotes expectation under $H_1: \beta_j \neq 0.$ 

The strategy used to define $\bar{T}^{t, \epsilon}(\mathcal{P})$ makes it difficult to derive analytical expressions for $r$ and $E_{H_1}\left\{\mathbb{T}_{t,\epsilon}(\mathcal{P}) \left| \bX_{\bD},\mathcal{R} \right.\right\}$ in terms of $(\alpha,M,a,\epsilon)$. However, since  $\bar{T}^{t,\epsilon}(\mathcal{P})$ is a function of random variables that are easy to generate numerically, it is trivial to use Monte Carlo simulation to provide accurate approximations of $r$ and $E_{H_1}\left\{\mathbb{T}_{t,\epsilon}(\mathcal{P}) \left| \bX_{\bD},\mathcal{R} \right.\right\}$. This allows us assess type II error rates for the test both asymptotically, which we study in this section, and in finite samples, which we study in Section \ref{Illustrations}.

%We acknowledge that Monte Carlo approximations do not correspond to theoretical results. Nevertheless, we consider valuable to provide a numerical---but accurate---approximation of the asymptotic behavior of the type II error in this section. 

As $n$ goes to infinity, $\bar{T}^{t,R}(\mathcal{P})$ converges in probability to $\sqrt{M}a$.  Hence, the asymptotic probability of type II error is approximately equal to the probability of the acceptance region $(-r,r)$ under a Laplace distribution with location and scale equal to $\sqrt{M}a$ and $2a/\sqrt{M}\epsilon$, respectively. This characterization immediately reveals that the probability of type II error for this test never equals zero.

Figure \ref{typeIIerror} displays a Monte Carlo approximation of the asymptotic probability of type II error associated with the test $\mathbb{T}_{t,\epsilon}(\mathcal{P})$ for different values of $(\alpha,a,M,\epsilon)$. 
%Although theoreti $\lim_{n\rightarrow \infty}E_{H_1}\left\{\mathbb{T}_{t,\epsilon}(\mathcal{P}) \left| \bX_{\bD},\mathcal{R} \right.\right\}$ cannot not be theoretically equal to zero, 
For most of the combinations of $(\alpha,M,a,\epsilon)$ studied here, it is possible to obtain an asymptotic probability of type II error that is close to zero, with $(\alpha=0.01, \epsilon = 0.1)$ being the lone exception. For any given $(\alpha,M,\epsilon)$, the asymptotic probability of type II error is almost constant as a function of $a$, that is, $\lim_{n\rightarrow \infty}E_{H_1}\left\{\mathbb{T}_{t,\epsilon}(\mathcal{P}) \left| \bX_{\bD},\mathcal{R} \right.\right\}$ decreases at a slow rate as $a$ increases.  The type II error probabilities also decrease as any one of $M$, $\alpha$, or $\epsilon$ increases, holding the others constant. We note that analysts can use this Monte Carlo approach to approximate the asymptotic probability of type II error for any combination of $(\alpha,M,a,\epsilon)$.  Thus, for example, for a fixed $(\alpha,\epsilon)$, analysts can determine values of $(M,a)$ that lead to a specified asymptotic probability of type II error. For example, for $\alpha= 0.05$, $\epsilon=1$, and $\lim_{n\rightarrow \infty}E_{H_1}\left\{\mathbb{T}_{t,\epsilon}(\mathcal{P}) \left| \bX_{\bD},\mathcal{R} \right.\right\}< 0.001$, we find that $M$ needs to be greater than 25 provided that $a$ is greater than one.

\begin{figure}[t]
\centering
\includegraphics[scale=0.75]{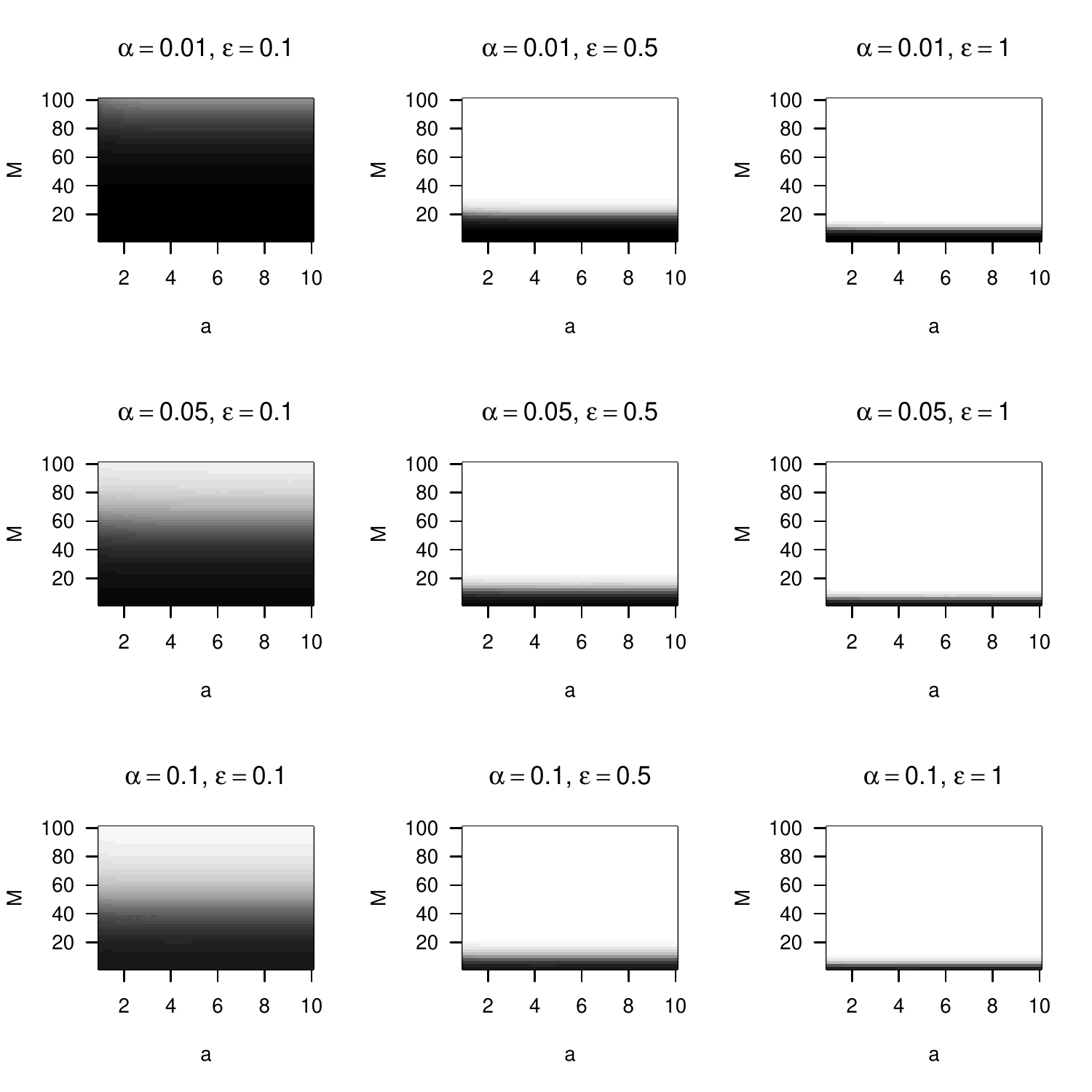}
\caption{Asymptotic probability of type II error associated with the test $\mathbb{T}_{t,\epsilon}(\mathcal{P})$ for different values of $(\alpha,a,M,\epsilon)$. Each asymptotic probability is computed using Monte Carlo approximation based on 10,000 realizations of $\mathbb{T}_{t,\epsilon}(\mathcal{P})$. Within any panel and combination $(M,a)$, black represents $\lim_{n\rightarrow \infty}E_{H_1}\left\{\mathbb{T}_{t,\epsilon}(\mathcal{P}) \left| \bX_{\bD},\mathcal{R} \right.\right\}  =1$ and white represents $\lim_{n\rightarrow \infty}E_{H_1}\left\{\mathbb{T}_{t,\epsilon}(\mathcal{P}) \left| \bX_{\bD},\mathcal{R} \right.\right\} = 0$, with lighter grays as the loss function approaches 0.   }
\label{typeIIerror}
\end{figure}

Although the Monte Carlo approach provides accurate and straightforward approximations, it is also instructive to characterize the asymptotic behavior of $E_{H_1}\left\{\mathbb{T}_{t,\epsilon}(\mathcal{P}) \left| \bX_{\bD},\mathcal{R} \right.\right\}$ mathematically under arbitrary choices of $(\alpha,a,M,\epsilon)$.   Theorem \ref{typeIIbound} provides an upper bound for $\lim_{n\rightarrow \infty}E_{H_1}\left\{\mathbb{T}_{t,\epsilon}(\mathcal{P}) \left| \bX_{\bD},\mathcal{R} \right.\right\}$. Analogous to Figure \ref{typeIIerror}, in the supplementary material we present a graphical representation of the upper bound. We observe that \eqref{eqtypeIIbound} is a sharp bound when $a>2$. For $1 \leq a \leq 2$ and some values of $M$, the bound in \eqref{eqtypeIIbound} is a moderately less precise, but still valid, bound for $\lim_{n\rightarrow \infty}E_{H_1}\left\{\mathbb{T}_{t,\epsilon}(\mathcal{P}) \left| \bX_{\bD},\mathcal{R} \right.\right\}$.  The proof of Theorem \ref{typeIIbound} is in the supplementary material.

\begin{theorem}\label{typeIIbound}
Under $H_1: \beta_j \neq 0$ and assumption {\bf A1},
\begin{eqnarray}\nonumber
\lim_{n \rightarrow \infty} E\left\{\mathbb{T}_{t,\epsilon} \left| \bX_{\bD},\mathcal{P} \right.\right\}
 & < & 
\mathbb{I}_{\{r^*<\sqrt{M}a\}}
\frac{1}{2}\left(\alpha^{-1-\epsilon\sqrt{M}/2a}-\alpha^{1+\epsilon\sqrt{M}/2a}\right)\exp\left(-\frac{\epsilon M}{2}\right)
\\\label{eqtypeIIbound} & &
\hspace*{-10mm}+ \ 
\mathbb{I}_{\{r^*>\sqrt{M}a\}}\left(
1-\frac{1}{2}\alpha^{1+\epsilon\sqrt{M}/2a}\left[\exp\left(\frac{\epsilon M}{2}\right)- \exp\left(-\frac{\epsilon M}{2}\right)\right]\right),
\end{eqnarray}
where $r^*=-log(\alpha)\left(2a/\epsilon\sqrt{M}+1\right)$.
\end{theorem}

\section{Empirical Illustrations}\label{Illustrations}

The results in Section \ref{ss:distance} suggest we should make $a$ as large as possible and $M$ as small as possible to ensure the proximity between $\bar{T}^{t,R}(\mathcal{P})$ and $T(\bD)$. However, the accuracy of the approximation of $T(\bD)$ is only part of the story. We need to add Laplace noise to protect privacy.  To maximize the usefulness of the privately-computed $t$-statistic $\bar{T}^{t,\epsilon}(\mathcal{P})$, we seek to add as little noise as possible while still satisfying differential privacy. This pushes us to make $a$ smaller rather than larger, and to make $M$ larger rather than smaller. How do we trade off accuracy in $\bar{T}^{t, R}(\mathcal{P})$ for reductions in variance of the Laplace noise? We have a partial answer this question from the asymptotic results in Section \ref{ss:power}, but, in practice, we have to analyze finite samples. Which choices of $(M,a)$ tend to offer higher accuracy for a given risk level $\epsilon$ with finite samples?  
%We examine these questions in the simulation studies of Section \ref{Illustrations}, and we present an approach for choosing values of $M$ and $a$ in Section \ref{ChoosingParams}.

In this section, we address this question 
%illustrate the performance of $T^{t,\epsilon}(\mathcal{P})$ in finite samples 
using simulation studies. 
%These studies also illuminate the trade offs in accuracy among different values of $M$, $a$, and $\epsilon$. 
In all simulations, we assume that $\hat{\beta}_j \sim {\rm N}(\beta_j, \Sigma_{j,j}(\bD))$ and $\hat{\beta}_{j,l} \sim {\rm N}(\beta_j, \Sigma_{j,j}(\bD_l))$ for any $j$ and $l$, i.e., we assume that {\bf A1} holds. We consider two scenarios, one where $M\bX_{\bD_l}^\top\bX_{\bD_l} \approx (\bX_{\bD}^\top\bX_{\bD})$, where $l=1,\dots,M$, i.e., where {\bf A2} holds, and the other where this is not necessarily the case. The scenarios are generated as follows.

In Scenario I, we work directly with the theoretical distributions of the $t$-statistics, without simulating and partitioning values of $\bD$. For an arbitrary regression coefficient $\beta$, let $\mu_T$ be the number of standard deviations that its value is from zero.  We consider multiple values of $\mu_T$ in the simulation.  For any $\mu_T$, we generate  $T$ and  $\bar{T}^{t,R}$  from their sampling distributions as follows,
\begin{align}\nonumber
T & \sim  {\rm N}(\mu_T,1) \\\label{T1_T2_1}
% \bar{T} & = \frac{1}{M}\sum_{l=1}^M \sqrt{M} Z_l, \\
\bar{T}^{t,R} & = \frac{1}{M}\sum_{l=1}^M \sqrt{M}\left(-a\mathbb{I}_{(-\infty,-a)}(Z_l) + Z_l\mathbb{I}_{[-a,a]}(Z_l) + a\mathbb{I}_{(a,\infty)}(Z_l)\right). 
\end{align}
We let $Z_l \overset{i.i.d.}{\sim} {\rm N}(\sqrt{M}\mu_T,1)$. We use the Laplace mechanism based on the appropriate global sensitivity values to generate $\bar{T}^{t, \epsilon}$.

Each $T$ and $\bar{T}^{t,R}$ is generated independently.  Generally, one would expect their values to be positively correlated when computed on some $\bD$.  However, generating them independently guarantees that {\bf A2} holds. Scenario I  provides lower bounds for cases where the $t$-statistics are positively correlated, since the  statistics should be more similar when positively correlated than when independent.

In Scenario II, we work with a subset of the March 2000 Current Population Survey (CPS) public use file comprising $n = 49,436$ heads of households with non-negative incomes.  This dataset was used by \cite{reiter;2005b} and \cite{chen;machanavajjahala;reiter;barrientos;2016},   among others.  In order to use realistic predictor distributions, we set $\bX_{\bD}$ to be an $n \times 25$ matrix of values derived from the CPS data.  Its columns include age in years, age squared, education (16 levels), marital status (7 levels), and sex (2 levels). We generate multiple sets of the response variable $Y$ from linear regressions on $\bX_{\bD}$, each using a different, pre-specified set of $\boldsymbol{\beta}= (\beta_0,\ldots,\beta_{24})^\top$, so as to control the importance of the regression coefficients.  For $j=0, \dots, 24$, let ${\Sigma_{j,j}(\bD)}$ be the $j$th diagonal element of ${\sigma}^2(\bX_{\bD}^\top\bX_{\bD})^{-1}$,  where $\sigma = 0.82$.   This value of $\sigma$ corresponds to the residual standard error of a linear regression fitted using the logarithm of income---one of the variables in the CPS data---as the response variable and $\bX_{\bD}$ as predictors.  To derive any one $\boldsymbol{\beta}= (\beta_0,\ldots,\beta_{24})^\top$, we set each $\beta_j = \mu_T \sqrt{\Sigma_{j,j}(\bD)}$ for some specified number of standard deviations $\mu_T$ from zero.  Using this $\boldsymbol{\beta}$, we simulate realizations of $T$ and $\bar{T}^{t,R}$ via the following steps.
\begin{itemize}
\item[i)] For $i=1, \dots, n$, generate $y_i$ from ${\rm N}(x_i^\top\boldsymbol{\beta},\sigma^2)$, where $x_i$ denotes the $i$th row-vector of $\bX_{\bD}$ and $\sigma =  0.82$.   
\item[ii)] Set $\bD=\{(x_i, y_i)\}_{i=1}^n$, and generate a random partition $\mathcal{P}=\{\bD_l\}_{l=1}^M$.
\item[iii)] Get a realization of $T = T(\bD)$ by computing the $t$-statistic of the $j$th regression coefficient estimated from $\bD$. 
\item[iv)] For $l=1, \dots, M$, set $Z_l = T(\bD_l)$ as the $t$-statistic of the $j$th regression coefficient obtained from the regression of $y$ on $\bX_{\bD_l}$.  Get realizations of $\bar{T}^{t,R}$ from (\ref{T1_T2_1}).
\end{itemize} 
We then add Laplace noise to generate $\bar{T}^{t, \epsilon}$.

For both scenarios, we let $\mu_T \in \{0,0.1,\ldots,1,2,\ldots,10\}$, let $M \in \{10, 25, 50, 75, 100\}$, let $a \in \{1, \ldots, 10\}$, and let $\epsilon \in\{0.5,1,2.5,5\}$.  We generate simulations for all possible combinations of $(\mu_T, M, a, \epsilon)$. For each combination $(\mu_T, M, a, \epsilon)$, we generate 100,000 realizations of $(T, \bar{T}^{t,R})$ for Scenario I and  1,000 realizations for Scenario II.  

We evaluate the significance tests based on $\bar{T}^{t, \epsilon}$ by comparing the power  at each value of $\mu_T$ to the power of the test based on $T$ at each corresponding value of $\mu_T$.  We evaluate the properties of the differentially private sign measures by comparing how often one can infer the correct sign of each $\beta_j$ from the corresponding $T$ and $\bar{T}^{t,\epsilon}$. We also compute the probability that a user makes the same decision about the significance level or sign when using $T$ and $\bar{T}^{t,\epsilon}$; we call these matching probabilities.

\subsection{Assessing inferences about  significance}\label{Illus:signif}

We study the significance properties of the $t$-statistics using the following quantities,
\begin{eqnarray*}
p_0({\rm t},\gamma) & = & P\left\{|T| < {\rm t} \left| \mu_T = \gamma, \bX_{\bD}  \right. \right\}, 
\\
p_{t,\epsilon}({\rm t},\gamma,M,a,\epsilon)  & =& P\left\{|\bar{T}^{t, \epsilon}| < {\rm t} \left| \mu_T  = \gamma, M, a, \epsilon, \bX_{\bD}, \mathcal{R} \right. \right\}.
\end{eqnarray*}
As a slight abuse of notation, we condition these probabilities on $(\gamma,M,a,\epsilon)$ to highlight that they parametrize the probabilities.

For a given significance level $\alpha$ and type II error rate $\lambda_0$, let $r_0$ and $q_0$ be positive constants such that $p_0(r_0,0)= 1-\alpha$ and $p_0(r_0,q_0)= \lambda_0$. Notice that $r_0$ is the $(1-\alpha/2)$th quantile of the distribution of $T$ when $\mu_T = 0$, i.e., $r_0$ is the critical value under the null hypothesis $H_0: \beta_j=0$ that ensures a confidence level of $1-\alpha$ for the test $\mathbb{T} = \mathbb{I}_{(-r_0,r_0)} (T)$. The value $q_0$ is the number of standard deviations from zero at which the test $\mathbb{T}$ reaches the desired Type II error $\lambda_0$, i.e., under $H_1: \beta_j = q_0 \sqrt{\Sigma_{j,j}(\bD)}$, the power of this test is equal to $1-\lambda_0$.  

Let $r$ and $\lambda(M, a, \epsilon)$ be positive constants such that $p_{t,\epsilon}(r,0,M,a,\epsilon)= 1-\alpha$ and $\lambda(M,a,\epsilon) = p_{t,\epsilon}(r,q_0,M,a,\epsilon)$.   For given values of $(M, a, \epsilon)$, $r$ is the critical value under $H_0$ that ensures a confidence level of $1-\alpha$ for the test $\mathbb{T}_{t,\epsilon} = \mathbb{I}_{(-r,r)}(\bar{T}^{t, \epsilon})$. The value $\lambda(M,a,\epsilon)$ corresponds to the type II error rate of $\mathbb{T}_{t,\epsilon}$ under $H_1: \beta_j = q_0 \sqrt{\Sigma_{j,j}(\bD)}$. 

To assess how similar the test $\mathbb{T}_{t,\epsilon}$ is to $\mathbb{T}$, we use the loss function,   
$$
L_{t,\epsilon}^{sig}(M,a,\epsilon) = \max(0,\lambda(M,a,\epsilon)-\lambda_0).
$$
For a given significance level $\alpha$, we say there is zero loss of power from using $\mathbb{T}_{t,\epsilon}$ when the type II error rate of $\mathbb{T}_{t,\epsilon}$ at $q_0$ is less than $\lambda_0$. Otherwise, we record the corresponding loss of power, $\lambda(M,a,\epsilon)-\lambda_0$.   For all results in this subsection, we set $\alpha = 0.05$ and $\lambda_0 = 0.2$.

We begin by examining the  values of $L_{t,\epsilon}^{sig}(M,a,\epsilon)$ at the different combinations of $(M,a)$ when $\epsilon = \infty$, i.e., no noise is added to $\bar{T}^{t,\epsilon}$.  To save space, we present these results in the supplementary material. 
% Figure \ref{Figure_power_1} in the online supplement,  MOVE GRAPH TO SUPPLEMENT
In Scenario I, the power for $\mathbb{T}_{t,\epsilon}$ and  $\mathbb{T}$ are almost identical for $a \geq 2$  for all values of $M$. This finding conforms with the theory in Section \ref{sec:theory}. 
%Surprisingly, even for small values of $a$ and $M$ when $\epsilon = \infty$,  the power values for $\mathbb{T}_{TR}$ and $\mathbb{T}$ are similar. Evidently, similar distributions of $\bar{T}^{TR}$ and $T$ can imply similar power functions of $\mathbb{T}_{TR}$ and $\mathbb{T}$, but not vice versa. 
The results with the CPS data in Scenario II follow a similar pattern except for $M=100$. This is expected since large values of $M$ weaken the validity of {\bf A2}.

%FELIPE:  PUT THIS PARAGRAPH IN THE ONLINE SUPPLEMENT WHERE YOU PRESENT THE FIGURE \ref{Figure_power_1}.  For both methods, in Scenario II the power differences are largest when $M=100$. Apparently, at $M=100$, {\bf A1} is unreasonable, reducing the power of the tests. The power of $\mathbb{T}_{RT}$ is more affected by the violation of {\bf A1} than the power of $\mathbb{T}_{TR}$.  
%SHOULD WE JUST REPORT THE FIGURE 2 IN THE SUPPLEMENT?  I AM WORRIED THAT WE HAVE TOO MANY FIGURES.

We next consider $\epsilon < \infty$, as needed to satisfy differential privacy.   Figure \ref{Figure_power_2} displays values of $L_{t,\epsilon}^{sig}(M,a,\epsilon)$ for $\epsilon \in \{0.5, 1.0, 2.5, 5.0\}$. 
In Scenario I, $L_{t,\epsilon}^{sig}(M,a,\epsilon)$ tends to be smaller when $M$ is large and $a$ is small, and largest when $M$ is small and $a$ is large. In other words, the discrepancy between the power of $\mathbb{T}_{t,\epsilon}$  and $\mathbb{T}$ is smaller when the global sensitivity is smallest. 
%This is quite different for $L_{RT}^{sig}$, which tends to be smallest when $M$ and $a$ are large, and largest when $M$ and $a$ are small. When $\epsilon < 5$, $L_{RT}^{sig} \leq L_{TR}^{sig}$  for most combinations of $(M,a)$, whereas the opposite occurs when $\epsilon = 5$.  
% with the exception of $M=100$.
%decrease of the global sensitivity of $\bar{T}^{RT}$ due to large values of $a$ is compensated with large values of $M$. In turn, the increase of the distance between the distributions of $\bar{T}^{RT}$ and $T$ due to large values of $M$ is compensated with large values of $a$. 
%The patterns observed for Scenario I are similar to those for Scenario II for almost of all combinations of $(M,a,\epsilon)$, with the exception of $M=100$. We also observe that $L_{RT}^{sig}<L_{TR}^{sig}$ for a large number of combinations of $(M,a)$. 
For values of $\epsilon > 1$, there is at least one combination of $(M,a)$ for $\mathbb{T}_{t,\epsilon}$ that provides almost no loss in  power.  For $\epsilon = 1$, with $\mathbb{T}_{t,\epsilon}$ we need $M \geq 50$ and $a \in \{1,2\}$  to experience only a small power loss. For $\epsilon = 0.5$, with $\mathbb{T}_{t,\epsilon}$  we still can achieve only modest power losses by using  $(M \geq 75, a = 1)$. 
% However,  with $\mathbb{T}_{TR}$ we have to set $(M = 75, a = 1)$, which arguably is too small a value for $a$ to provide useful feedback on statistical significance.
As expected, the power of $\mathbb{T}_{t,\epsilon}$ is strongly influenced by $\epsilon$. As $\epsilon$ gets smaller, so does the power.  The reduction in power for small values of $\epsilon$ is the price to pay for strengthening the privacy guarantee. However, we emphasize that $\mathbb{T}_{t,\epsilon}$ is still a valid test when $\epsilon$ is small, in that it gives the correct type I error rates regardless of the value of $\epsilon$.
Finally, the result patterns obtained under Scenario II generally match those obtained under Scenario I except for $M=100$. For this value of $M$, we notice a loss of power for some of the regression coefficients. We attribute this discrepancy to the fact that large values of $M$ weaken the validity of {\bf A2}.

\begin{figure}[t]
\centering
\includegraphics[scale=0.6]{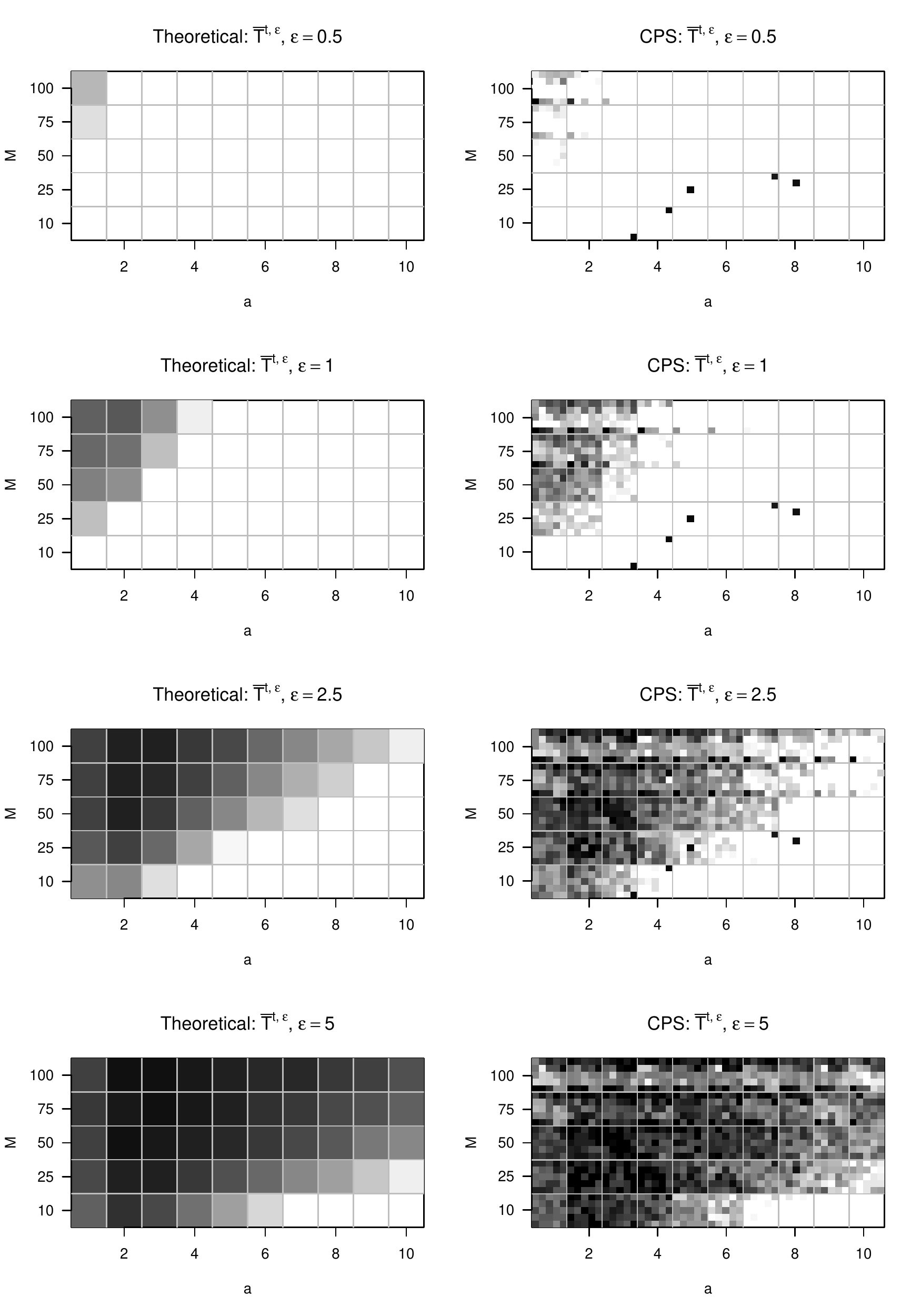}
\caption{Values of $L_{t,\epsilon}^{sig}(M,a,\epsilon)$  associated with $\bar{T}^{t,\epsilon}$ for significance at different combinations of $(M, a, \epsilon)$ with $\alpha = 0.05$, and $\lambda_0 = 0.2$.  Left and right panels show the results for Scenario I and II, respectively. For Scenario I, each cell represents the average of 100,000 runs at that $(M,a)$.  For Scenario II, each $(M,a)$ cell comprises a $5 \times 5$ array of sub-cells representing the results for the $25$ coefficients in the CPS regression. Individual sub-cell results are averages of 1,000 runs. 
Within any cell or sub-cell, black represents $L_{t,\epsilon}^{sig}(M,a,\epsilon) = 0$ and white represents $L_{t,\epsilon}^{sig}(M,a,\epsilon) \geq 0.25$, with lighter grays as the loss function approaches 0.25.}
%For Scenario I panels (top panels), each combination $(M,a)$ is represented as a cell in gray-scale, where $L_k^{sig}(M,a,\infty)=0$ is black-colored and $L_k^{sig}(M,a,\infty) \geq 0.25$ is white-colored. For Scenario II panels (bottom panels), each combination $(M,a)$ is represented as a cell comprising an array of $5 \times 5$ sub-cells. Each $5 \times 5$ array contains $25$ sub-cells representing the results obtained for the $25$ coefficients. Sub-Cells are colored using a gray-scale, where $L_k^{sig}(M,a,\infty)=0$ is black-colored and $L_k^{sig}(M,a,\infty) \geq 0.25$ is white-colored.  LOOKS LIKE MINUS SIGNS AT TOP OF GRAPHS.  MAYBE USE : AFTER THEORETICAL AND , AFTER T}
\label{Figure_power_2}
\end{figure}

\subsection{Assessing inferences about signs of coefficients}\label{Illus:sign}
Since the sign of $\mu_T$ and $\beta_j$ is the same, we restrict our analysis to the sign of $\mu_T$ only. Because the Laplace and Student-$t$ distributions are symmetric, we only consider the case where $\mu_T \geq 0$. 
We study the sign of the statistics using the following quantities,
\begin{align*}
s_0(\gamma) & = P\left\{{\rm sign}(T) = {\rm sign}(\mu_T) \left| \mu_T = \gamma, \bX_{\bD}\right. \right\},
\\
s_{t,\epsilon}(\gamma,M,a,\epsilon)  & = P\left\{{\rm sign}\left(\bar{T}^{t,\epsilon}\right) =  {\rm sign}(\mu_T) \left| \mu_T  = \gamma, M, a, \epsilon, \bX_{\bD}, \mathcal{R} \right. \right\}.  
\end{align*}
These represent  probabilities that the $t$-statistics have the same sign as $\beta_j$  when the value of $\beta_j$ is $\gamma$ standard deviations from zero. For a given probability $\alpha_0$, let $\mu_0$ be a positive constant such that $s_0(\mu_0) = \alpha_0$, i.e., $\mu_0$ is the number of standard deviations at which $T(\bD)$ has the same sign of $\beta_j$ with a probability equal to $\alpha_0$. To assess how similar $\alpha_0$ is to $s_{t,\epsilon}(\mu_0,M,a,\epsilon)$, we use the loss function   
$$
L_{t,\epsilon}^{sgn}(M,a,\epsilon) = \max(0,\alpha_0-s_{t,\epsilon}(\mu_0,M,a,\epsilon)). 
$$
We say that $\bar{T}^{t,\epsilon}$ and $T(\bD)$ result in similar inferences about the sign of $\beta_j$ when \linebreak $\alpha_0>s_{t,\epsilon}(\mu_0,M,a,\epsilon)$. For all results in this subsection, we set $\alpha_0 = 0.95$.

We again begin with the values of $L_{t,\epsilon}^{sgn}(M,a,\epsilon)$ for $\epsilon=\infty$; results are displayed in the supplementary material.  
%Figure \ref{Figure_sign_1} shows the values of $L_k^{sgn}(M,a,\epsilon)$ for the simulation scenarios.  Here, we set $\epsilon=\infty$ (no noise is added to  $\bar{T}^{TR}$ and $\bar{T}^{RT}$). 
Across scenarios, inferences about the sign of $\beta_j$ based on $\bar{T}^{t,\epsilon=\infty}$ are quite accurate  
%In both scenarios, when $\epsilon = \infty$,  the %values of the loss function are very close to zero 
for all combinations of $(M,a)$. 
Figure \ref{Figure_sign_2} displays results for $L_{t,\epsilon}^{sgn}(M,a,\epsilon)$ when $\epsilon \in \{0.5, 1.0, 2.5, 5.0\}$. 
In general, when $\epsilon \leq 2.5$,  $\bar{T}^{t,\epsilon}$ offers many combinations of $(M, a)$ that result in accurate inferences about the sign of $\beta_j$, especially when $a$ is small and $M$ is large. We also find combinations of $(M,a)$ that result in accurate inferences about the sign even when $\epsilon = 0.5$, in particular when $M \geq 50$ and $a \in \{1,2\}$.  When $\epsilon = 5$, the results for  $\bar{T}^{t,\epsilon}$ are practically indistinguishable for almost all combinations of $(M,a)$. These general findings hold for both scenarios, so that  $\bar{T}^{t,\epsilon}$ provides a differentially private mechanism to release the sign of $\beta_j$ that seems robust against violations of {\bf A2}.

%
% Scenario$L_{TR}^{sgn}$ takes on smaller values when $M$ is large and $a$ is small, and takes on larger values when $M$ is small and $a$ is large. For a %given combination of $(M,a)$, $L_{k}^{sgn}$ increases as $\epsilon$ decreases. Although the patterns of $L_{k}^{sng}$ and $L_{k}^{sig}$ as a function of %$(M,a,\epsilon)$ are similar, their magnitudes are different. In general, $L_{k}^{sgn} < L_{k}^{sig}$. This leads us to conclude that we need to spend le%ss $\epsilon$ in those cases where the goal is to infer about the sign and not about the significance. In fact, for $\epsilon>2.5$, the results obtained %for $\bar{T}^{RT,\epsilon}$ show a similar performance to that of $\bar{T}^{RT}$.  For small values of $\epsilon$, the overall performance of $\bar{T}^{R%T,\epsilon}$ for inferring the sign of $\beta_j$ is better than that of $\bar{T}^{TR,\epsilon}$.  

\begin{figure}[t]
\centering
\includegraphics[scale=0.6]{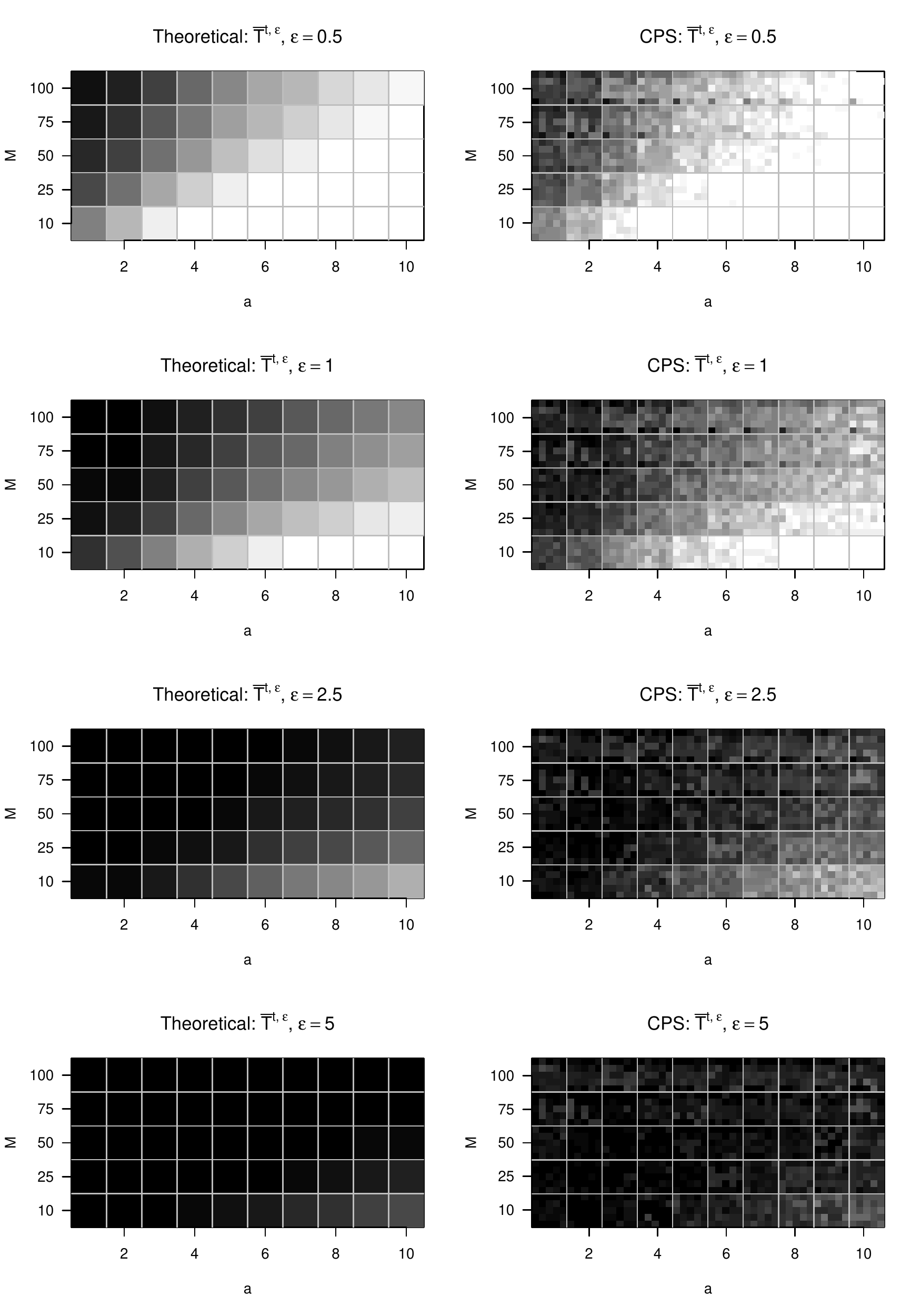}
\caption{Values of  $L_{t,\epsilon}^{sgn}(M,a,\epsilon)$ at different combinations of $(M, a, \epsilon)$ with $\alpha_0 = 0.95$. Left and right panels show the results for Scenario I and II, respectively. For Scenario I, each cell represents the average of 100,000 runs at that $(M,a)$.  For Scenario II, each $(M,a)$ cell comprises a $5 \times 5$ array of sub-cells representing the results for the $25$ coefficients in the CPS regression. Individual sub-cell results are averages of 1,000 runs. Within any cell or sub-cell, black represents $L_{t,\epsilon}^{sgn}(M,a,\epsilon) = 0$ and white represents $L_{t,\epsilon}^{sgn}(M,a,\epsilon) \geq 0.25$, with lighter grays as the loss function approaches 0.25.} \label{Figure_sign_2}
\end{figure}

\subsection{Matching probabilities}\label{matching}

We define matching probabilities as the probability that ${\rm sign}(T) = {\rm sign}(\bar{T}^{t, \epsilon})$ and the probability that $\mathbb{T}_{t,\epsilon}=\mathbb{T}$. We assess these probabilities using the quantities,
\begin{align*}
m_{t, \epsilon}^{sig}(\mu_T,\epsilon) &= \min_{M,a}P\left\{\mathbb{T}_{t, \epsilon}=\mathbb{T}|\mu_T,M,a,\epsilon, \bX_{\bD}, \mathcal{R} \right\}, \\
m_{t, \epsilon}^{sgn}(\mu_T,\epsilon) &= \min_{M,a}P\left\{{\rm sign}(T) = {\rm sign}(\bar{T}^{t, \epsilon})|\mu_T,M,a,\epsilon, \bX_{\bD}, \mathcal{R} \right\}, 
\end{align*}
where each minimum is over all possible combinations of  $M \in \{10,25,50,75,100\}$ and $a \in \{1,2,\ldots,10\}$.  As we minimize over $(M,a)$, these metrics represent the worst case matching probabilities for these simulations. In the supplementary material, we present analogous figures showing the maximum values of the matching probabilities, which represent the best case matching probabilities for these simulations.

Figure \ref{Figure_matching_prob_significance} displays the values of $m_{t, \epsilon}^{sig}(\mu_T,\epsilon)$ for different values of $\mu_T$ and $\epsilon$. The results in Scenario I and II follow similar patterns so we describe them simultaneously.   Under the null hypothesis, i.e., $\mu_T=0$, values of $m_{t, \epsilon}^{sig}$ are greater than $0.85$.  When $\mu_T$ is large enough, values of $m_{t, \epsilon}^{sig}$ are close to one.   The value of $\mu_T$ at which $m_{t, \epsilon}^{sig}$ is close to one is inversely related to the value of $\epsilon$. For example, when $\epsilon =5$ and $\mu_T>5$, then $m_{t, \epsilon}^{sig} \approx 1$. However, when $\epsilon =0.5$, $m_{t, \epsilon}^{sig} < 0.5$ for all values of $\mu_T$.  Values of $m_{t, \epsilon}^{sig}$ tend to be small when $\mu_T \in [r_0,r]$, where $r_0$ and $r$ are the critical values associated with $\mathbb{T}$ and $\mathbb{T}_{t, \epsilon}$, respectively. 
%This is due to the fact that $\mathbb{T}_k$ and $\mathbb{T}$ are inaccurate when $\mu_T$ is close to their corresponding critical values. Hence, 
Because $r$ increases as $\epsilon$ decreases, the range of values of $\mu_T$ where $m_{t, \epsilon}^{sig}$ is small becomes wider as $\epsilon$ decreases.
% Thus, $\mathbb{T}_{RT}$ provides better results than $\mathbb{T}_{TR}$.  In fact, when $\epsilon>2.5$, results based on $\bar{T}^{RT,\epsilon}$ are very similar to those based on $\bar{T}^{RT}$.

Figure \ref{Figure_matching_prob_significance} also summarizes the results for $m_{t,\epsilon}^{sgn}(\mu_T,\epsilon)$.   As expected, in both scenarios, increases in $\mu_T$ correspond to increases in the matching probability. The rate at which  $m_{t,\epsilon}^{sgn}$ increases as a function of $\mu_T$ depends on $\epsilon$: the larger the $\epsilon$, the faster the rate. In both scenarios, we observe a high matching probability ($m_{t,\epsilon}^{sgn}>0.9$) when $\mu_T \geq 3$ and $\epsilon\geq 2.5$. When $\mu_T<3$,  $m_{t,\epsilon}^{sgn}$ ranges from $0.5$ to $0.9$ for almost all $\epsilon$. It reaches a minimum value when $\mu_T=0$, regardless of the value of $\epsilon$; however, when $\mu_T=0$, matching the sign of  $T$ and  $\bar{T}^{t,\epsilon}$ arguably is not important for interpretations.

\begin{figure}[t]
\centering
\includegraphics[scale=0.33]{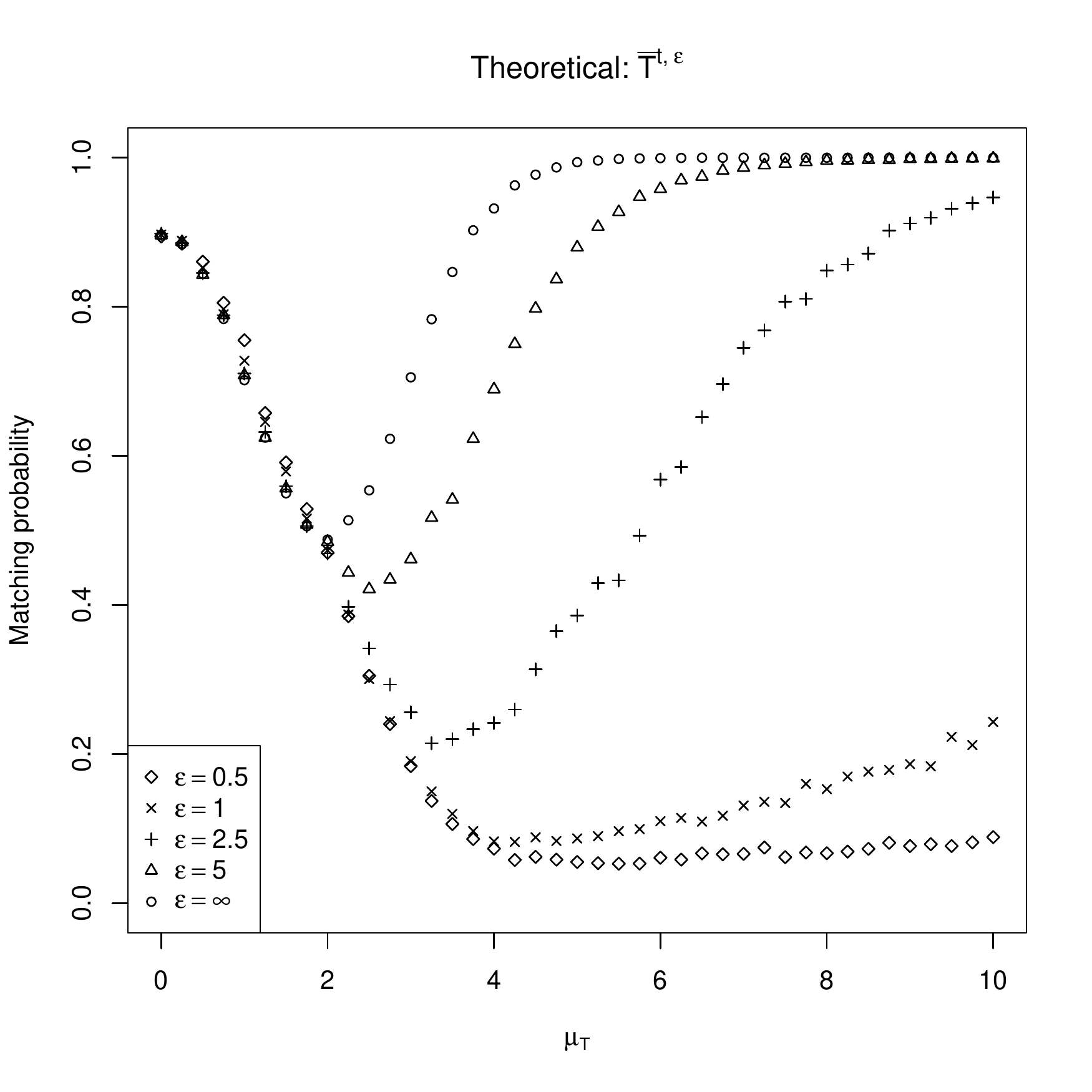}
\includegraphics[scale=0.33]{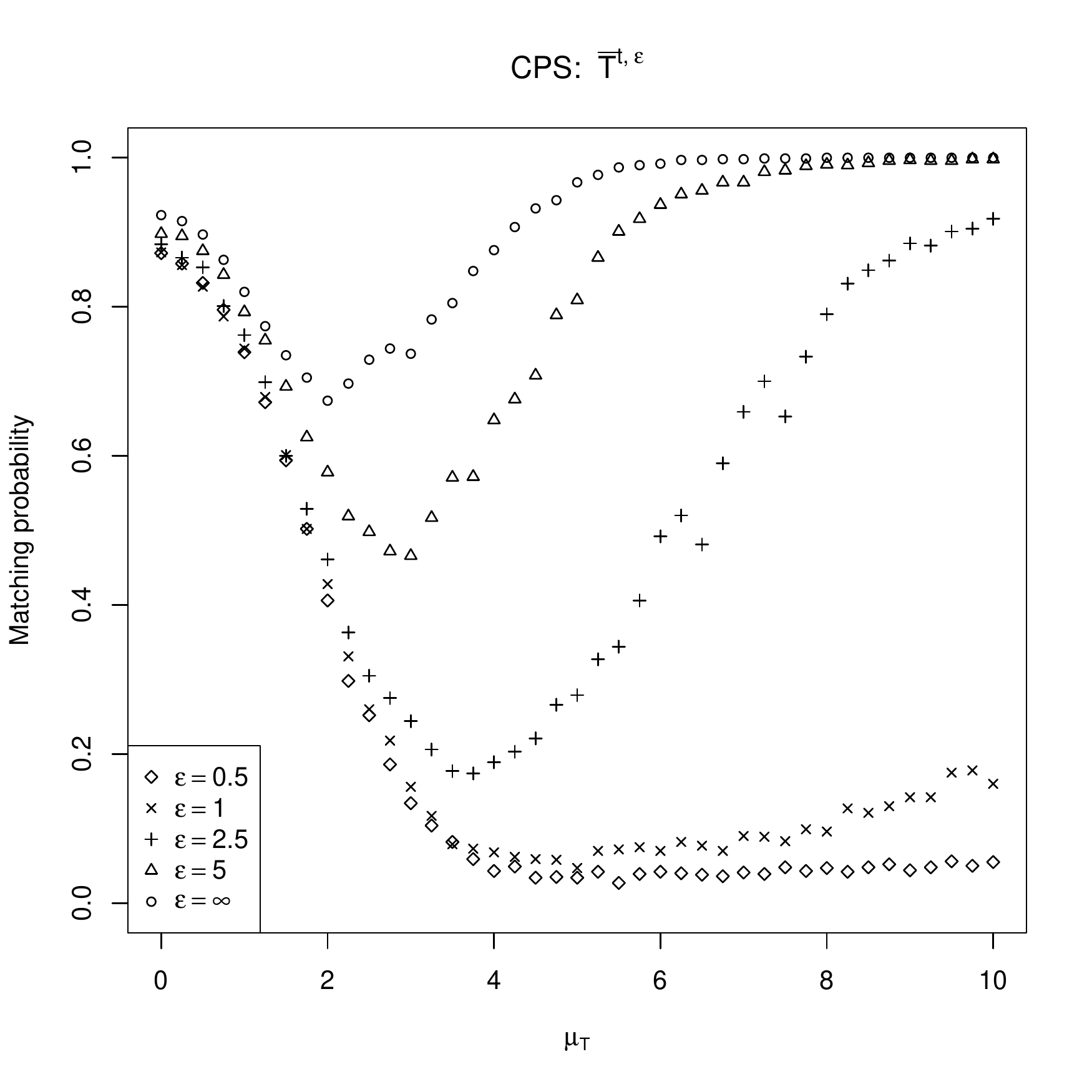}
\includegraphics[scale=0.33]{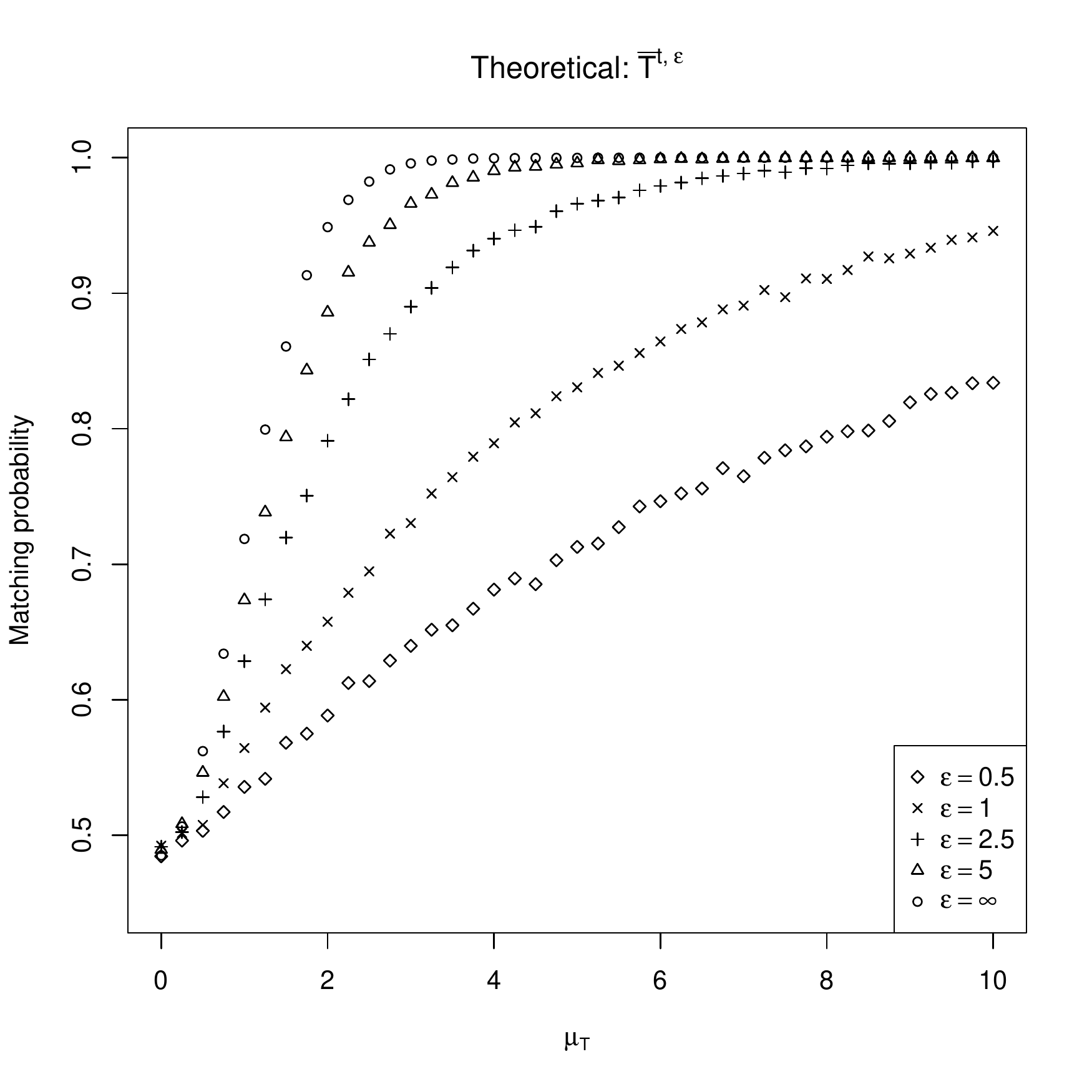}
\includegraphics[scale=0.33]{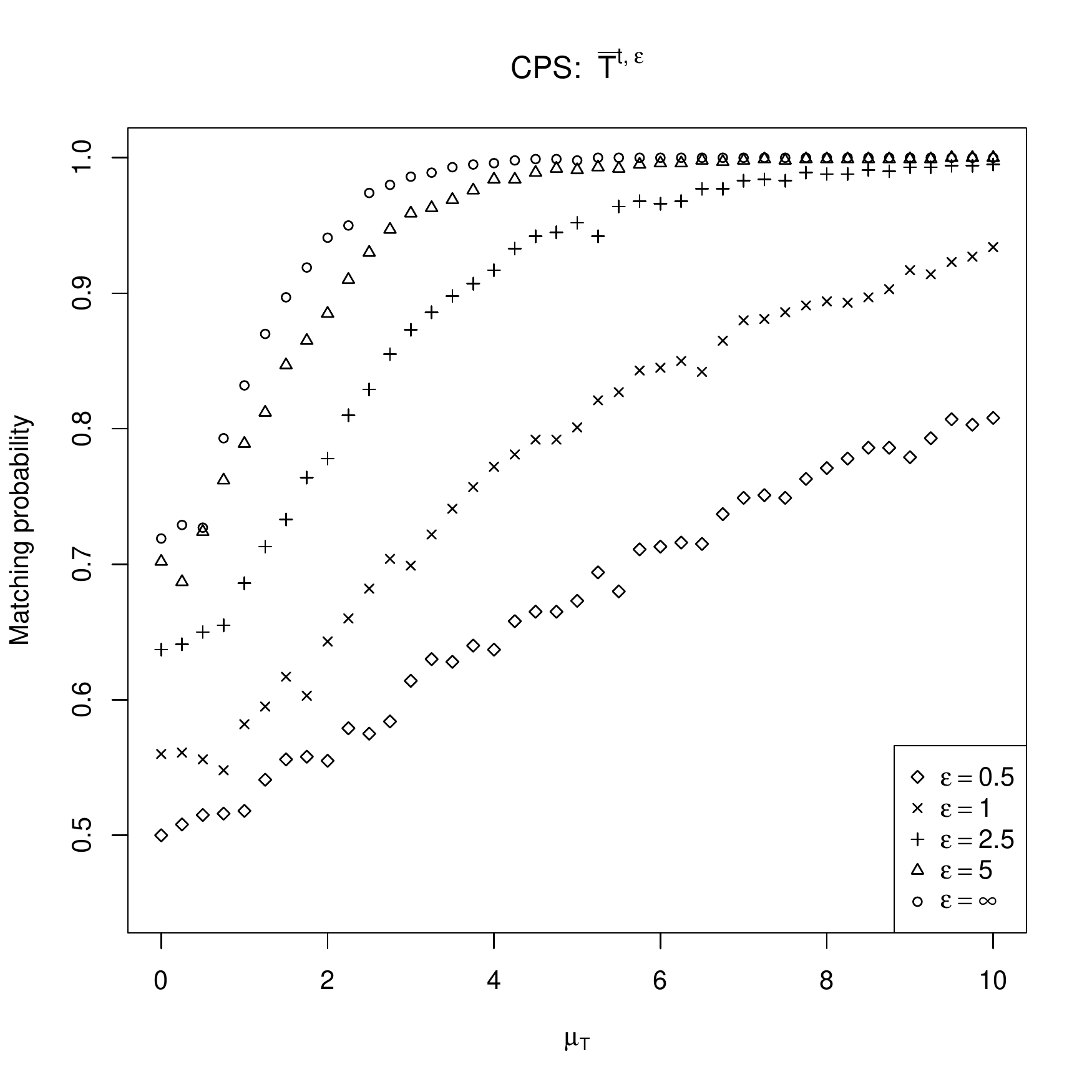}\caption{Values of  $m_{t,\epsilon}^{sig}(\mu_T,\epsilon)$ and $m_{t,\epsilon}^{sgn}(\mu_T,\epsilon)$ at different combinations of $\mu_T$ and $\epsilon$ with $\alpha = 0.05$.   Left and right panels show Scenario I and II, respectively. Top and bottom panels show the results for $m_{t,\epsilon}^{sig}(\mu_T,\epsilon)$ and $m_{t,\epsilon}^{sgn}(\mu_T,\epsilon)$,  respectively. For Scenario II, each point presents $\mu_T$ and the minimum value of $m_k^{sig}(\mu_T,\epsilon)$ taken over the $25$ regression coefficients.}
\label{Figure_matching_prob_significance}
\end{figure}

\section{Choosing $M$ and $a$ Without Additional Privacy Loss}\label{ChoosingParams}

To use these differentially private test statistics, the data producer or, when permitted in a verification server, the analyst first fixes the desired privacy level $\epsilon$ and then must select values for $M$ and $a$. Here, we consider an analyst who does not get to choose $\epsilon$ but does get to choose $(M,a)$.  Analysts who get to choose $\epsilon$, e.g., when allocating a total privacy budget across multiple queries, could repeat the approach described here with different values of $\epsilon$.
%privacy budget that the analyst wants to spend on querying $\bar{T}^{t,\epsilon}$. 

%For the choice of $(M, a)$, the analyst must balance competing objectives.  We want $a$ to be as small as possible to ensure low sensitivity, yet also not to be so small that the truncation results in a poor approximation of $T(\bD)$ or obscures the significance of the coefficient.  We want $M$ to be as small as possible to make {\bf A1} and {\bf A2} as believable as possible, yet also $M$ to be large to reduce the sensitivity.

In this section, we present a three step approach for selecting values of $(M, a)$ that does not incur additional privacy loss.  We illustrate these steps with a regression analysis of the CPS data, using the same $\bX_{\bD}$ as before and the reported values of household income on a logarithmic scale as the response variable. This regression fits reasonably well without obvious violations of the assumption of  i.i.d. Gaussian errors. We present the methodology for $\bar{T}^{t,\epsilon}$ with $\epsilon = 1.5$. The three-step approach also can be used for the sign.

\emph{Step 1: Fix an upper bound for $L_{t,\epsilon}^{sig}(M,a,\epsilon)$}. To begin, the user specifies an upper bound for $L_{t,\epsilon}^{sig}(M,a,\epsilon)$ for their desired significance level $\alpha$ and type II error rate $\lambda_0$.   By fixing $\alpha$ and $\lambda_0$, it is implied that there exists $q_0>0$ such that the power of $T$ equals $(1-\lambda_0)$  when $\beta_j$ is $q_0$ standard deviations from zero. Thus, the bound for $L_{t,\epsilon}^{sig}(M,a,\epsilon)$ represents the loss of power that a user is willing to accept by using $\bar{T}^{t, \epsilon}$ instead of using $T$ when $\beta_j$ is  $q_0$ standard deviations from zero.  In our illustrative example, we set $\alpha = 0.05$ and $\lambda_0 = 0.2$, and fix an upper bound of $0.1$.

\emph{Step 2: Simulate values of $L_{t,\epsilon}^{sig}(M,a,\epsilon)$ for choices of $M$ and $a$}. The user can simulate values of  $L_{t,\epsilon}^{sig}(M,a,\epsilon)$ for different combinations of $(M, a)$ using the strategy described in Scenario I of Section \ref{Illustrations}.  As an example, Table \ref{loss_function_epsilon1.5} displays the values of $L_{t,\epsilon}^{sig}(M,a,\epsilon)$ for different combinations of $(M, a)$.  Of course, this table is not comprehensive; users can create different tables with different combinations of $(M,a)$.  Importantly, we do not compute the entries in Table \ref{loss_function_epsilon1.5} using the CPS data; otherwise, the results would leak information about the confidential data. When auxiliary data are available, such as synthetic data, analysts could base the tables off these auxiliary data using the strategy described in Scenario II of Section \ref{Illustrations}. 

\begin{table}[t]
\begin{center}
\begin{tabular}{l|ccccc}
& \multicolumn{5}{|c}{$M$} \\
$a$  & 10 & 25 & 50 & 75 & 100 \\ \hline
1	&	0.13	&	{\bf 0.05}	&	{\bf 0.02}	&	{\bf 0.01}	&	{\bf 0.01}\\
2	&	0.17	&	{\bf 0.05}	&	{\bf 0.01}	&	{\bf 0.01}	&	{\bf 0.00}\\
3	&	0.32	&	0.11	&	{\bf 0.04}	&	{\bf 0.02}	&	{\bf 0.01}\\
4	&	0.51	&	0.22	&	{\bf 0.10}	&	{\bf 0.06}	&	{\bf 0.04}\\
5	&	0.65	&	0.34	&	0.16	&	{\bf 0.10}	&	{\bf 0.07}\\
6	&	0.74	&	0.47	&	0.25	&	0.16	&	0.12\\
7	&	0.79	&	0.58	&	0.34	&	0.22	&	0.16\\
8	&	0.82	&	0.66	&	0.43	&	0.30	&	0.21\\
9	&	0.84	&	0.72	&	0.51	&	0.37	&	0.27\\
10	&	0.86	&	0.77	&	0.59	&	0.44	&	0.34
\end{tabular}
\end{center}
\caption{\label{loss_function_epsilon1.5} Theoretical values of $L_{t,\epsilon}^{sgn}(M,a,1.5)$ at different combinations of $(a,M)$, $\alpha = 0.05$, and $\lambda_0 = 0.2$. Values in bold satisfy the condition $L_{t,\epsilon}^{sgn}(M,a,1.5)<0.10$.}
\end{table}

\emph{Step 3: Choose $(M, a)$}.  
The user considers all values of $(M,a)$ corresponding to values of $L_{t,\epsilon}^{sig}(M,a,\epsilon)$  below the fixed upper bound. When no combination of $(M, a)$ satisfies this condition, the user has to sacrifice accuracy and increase their error tolerance.
Once possible solutions exist, we recommend that the user choose the smallest value of $M$ from these solutions. Smaller values of $M$ result in larger sample sizes within the partitions, lessening the possibility that {\bf A1} or {\bf A2} are invalid.  
In our example, based on Table \ref{loss_function_epsilon1.5}, the user should choose the smallest values of $M$ for which $L_{t,\epsilon}^{sig}(M,a,\epsilon=1.5)$ is below 0.1, which is $M = 25$.  The user then chooses the value of $a$ that minimizes $L_{t,\epsilon}^{sig}(M=25,a,\epsilon=1.5)$.  From Table \ref{loss_function_epsilon1.5}, $L_{t,\epsilon}^{sig}(M=25,a,\epsilon=1.5)$ reaches its minimum at $0.05$ for $\mathbb{T}_{t,\epsilon}$ and $a \in \{1,2\}$. In theory, any of these values of $a$ should work. In this case, we recommend $a=2$ since this truncation level has the least effect on the approximation to $T(\bD)$. 

We now illustrate this method of choosing $(M, a)$ on the CPS data.  Using $(M=25, a = 2)$, we compute the differentially private p-value and sign using Algorithms 1 and 2; results for all 25 coefficients are displayed in the supplementary material.  
%he p-values and signs associated with five queries of $\bar{T}^{t, \epsilon}$ for all  $25$ coefficients in the regression of income on $\bX_{\bD}$ in the CPS data.  
The p-values from the differentially private significance test and the test based on the confidential data $\bD$  agree substantially, resulting in essentially the same conclusions about the significance for most of the regression coefficients. We observe discrepancies when the p-values based on $\bD$ provide weak evidence against the null. Regarding the sign, the outputs from the differentially private algorithm agree with the observed signs for almost all coefficients, except some with large p-values.  When p-values for the test of $H_0: \beta_j =0$  are quite large, changes in sign arguably are inconsequential.

As an additional illustration, we repeat the model selection and estimation using $\epsilon= 0.5$ and upper bound  for $L_{t,\epsilon}^{sig}(M,a,\epsilon)$ equal to $0.2$; detailed results are in the supplementary material. The $(M, a)$ selection algorithm suggests that we set $(M=100, a=1)$.  For estimating the signs, the differentially private algorithms continue to be effective: the privately-computed and observed signs systematically agree on all but three of the coefficients with significant p-values.  For p-values, we see  different interpretations of the significance for 3 out of 25 coefficients.
% in at least 3 out of the 5 queries.  %These tend to be coefficients with p-values between .02 and .10 in the confidential-data regression. 
We conjecture why this occurs when discussing the results in the supplementary material. 
%The degradation in quality of the verification is a consequence of demanding  higher privacy level.  
Of course, interpretations about the quality of the results for different $(M, a)$ and $\epsilon$ are data-specific, and we expect lower reductions in data quality for $\bD$ with larger sample size. 

\section{Concluding Remarks}\label{conclusion}

The methods described here allow data producers to provide $\epsilon$-differentially private answers to queries about the statistical significance and signs of coefficients in linear regression models.  Further, the strategy from Section \ref{ChoosingParams} provides users with a principled way to choose $(M, a)$.
% provides analto design high-powered tests,  so that they can have high confidence that they will arrive at valid conclusions from the private statistics.
We expect that these methods can be applied to significance testing with other regression models.  The key assumptions are {\bf A1} and {\bf A2}, which  are reasonable for many models and data settings.  

The  methods extend trivially to two contexts that are not discussed in previous sections. First, we can test general null hypotheses, $H_0: \beta_j = b$ with $b \in \mathbb{R}$. We simply define the $t$-statistic as $(\hat{\beta}_j-b)/ \sqrt{\hat{\Sigma}_{j,j}}$, and apply the algorithms without modification.  One application of this extension, which we leave to future investigation, is to set $b$ equal to the output of a privately computed regression coefficient $\tilde{\beta_j}$.  The test statistic then could be interpreted as a noisy, truncated estimate of the number of standard errors  $\tilde{\beta_j}$ is from $\hat{\beta}_j$.
%We also can test  $H_0: \beta_j = b$  against an alternative of the form $H_1: \beta_j < b$. Under either of these alternative hypotheses, we need to slightly modify the expression for the p-value in Algorithm 2, Figure \ref{Figure_algorithms}. Specifically, for $H_1: \beta_j > b$, p.value is equal to ``the proportion of $s_k$ such that $s_k>t_1$.'' Analogously, we can compute the p.value when $H_1: \beta_j > b$. 
Second, we can test the significance of the intercept, i.e., $H_0: \beta_0 = b$. This is equivalent to a $\epsilon$-DP, one sample t-test for a single mean. Unlike other tests for means, this significance test does not presume known bounds on data values or population means.  Nonetheless, it would be interesting to compare the power of the various tests for single means to see when each offers the best performance.

The algorithm described here provides results for one $\beta_j$ at a time. If analysts are given a finite privacy budget, they must spend part of that budget for each coefficient they wish to verify.  Thus, a key area for research is to develop differentially private algorithms that allow queries for multiple test statistics without burning through the privacy budget too quickly.

%Although we recognize the subjectivity behind simulating sets of predictors, we would like to make some comments about other simulations we conducted that are not included in this paper. We conducted those simulations with the purpose  to evaluate how assumption {\bf A2)} behaves for different types of predictors. Our simulated datasets included the same number of data points and predictors as the CPS dataset. We generated our simulations from a multivariate Gaussian distribution with zero mean and different specifications for the covariance matrix. Throughout all the simulated datasets we observed the same pattern with respect to our theoretical results. However, contrary to what we observed with the CPS dataset, the simulations not included in this paper do not show that assumption {\bf A2)} gets weaker for $M=100$. This shows that it is not easy to derive a rule of thumb for capping $M$ as a function of the sample size and number of predictors such that assumption {\bf A2)} still ``holds'' when its value is below that cap. It seems that the value of $M$ that makes weaker assumption {\bf A2)} not only depends on the sample size and number of predictors, but also on the type of predictors (categorical or numerical) and their correlation structure. 

%Remember that the user would need to use $\epsilon=1.5$ for each query and for each regression coefficient. This means if the user submits a query asking for the sign and significance of the $25$ coefficients, s/he would need to spend $1.5 \times 25$ units of $\epsilon$.

\section*{Supplementary Material}

The online supplementary materials include a graphical representation of the upper bound provided in Theorem \ref{typeIIbound}, and the proof of this theorem. They also include additional plots from simulations in Section \ref{Illus:signif} and Section \ref{Illus:sign} when $\epsilon=\infty$, and in Section \ref{matching} when the maximum of the matching probabilities is taken over $(M,a)$ for fixed values of $\epsilon$. Finally, the include the p-values and signs for the coefficients in the two examples considered in Section \ref{ChoosingParams}.

\section*{Acknowledgments}
This work is supported by grants from the National Science Foundation (ACI 1443014 and SES 1131897) and the Alfred P. Sloan Foundation (G-2-15-20166003).

%\bibliographystyle{plainnat}
%\bibliography{ref}

\end{document}